%% file: Packet_Loss_Control_Paper_V17.tex
\documentclass[letterpaper, 10 pt, conference]{ieeeconf}

\usepackage{cite}
\usepackage{amsmath,amssymb,amsfonts}
\usepackage{algorithmic}
\usepackage{graphicx}
\usepackage{textcomp}
\usepackage{xcolor}
\def\BibTeX{{\rm B\kern-.05em{\sc i\kern-.025em b}\kern-.08em
		T\kern-.1667em\lower.7ex\hbox{E}\kern-.125emX}}

\usepackage{graphics} % for pdf, bitmapped graphics files
\usepackage{epsfig} % for postscript graphics files
\usepackage{mathptmx} % assumes new font selection scheme installed
\usepackage{times} % assumes new font selection scheme installed
\usepackage{graphicx}  
\usepackage{bbold}
\usepackage{float}
\usepackage{circuitikz}
\usepackage{tikz}
\usepackage{caption}
\usepackage{dsfont}
\usepackage{subcaption}
\usepackage{phoenician}
\usepackage{caption}
\usepackage{bm}

\usepackage{array}
\usepackage{stfloats}
\usepackage{hhline}

\input{macros}

\newcommand{\be}{\begin{equation}}
\newcommand{\ee}{\end{equation}}
\newcommand{\beqn}{\begin{eqnarray}}
\newcommand{\eeqn}{\end{eqnarray}}
\newcommand{\ba}{\begin{array}}
	\newcommand{\ea}{\end{array}}

\newcommand{\R}{\mathds{R}}
\newcommand{\N}{\mathds{N}}
\newcommand{\Vs}{\ \!\! \Upsilon_k\!\! \ }
\newcommand{\Vsb}{\ \!\! \overline{\Upsilon}\!\! \ }

\newcommand{\nquad}{\mkern-18mu}
\newcommand{\nqquad}{\mkern-36mu}

\newcommand{\tr}{\textnormal{tr}}
\newcommand{\sminus}{\scalebox{0.5}[1.0]{\( - \)}}
\newcommand{\minus}{\scalebox{0.75}[1.0]{\( - \)}}
\newcommand{\splus}{\scalebox{0.5}[0.5]{\( + \)}}

\newcommand{\vR}{\left[0,1\right]}
\newcommand{\betam}{\bm{\beta}}

\newtheorem{innercustomgeneric}{\customgenericname}
\providecommand{\customgenericname}{}
\newcommand{\newcustomtheorem}[2]{%
	\newenvironment{#1}[1]
	{%
		\renewcommand\customgenericname{#2}%
		\renewcommand\theinnercustomgeneric{##1}%
		\innercustomgeneric
	}
	{\endinnercustomgeneric}
}

\newcustomtheorem{Appthm}{Theorem}
\newcustomtheorem{Applemma}{Lemma}
\newcustomtheorem{Appremark}{Remark}

\DeclareMathOperator*{\argmin}{arg\,min}

\DeclareMathAlphabet{\pazocal}{OMS}{zplm}{m}{n}

\newtheorem{theorem}{Theorem}
\newtheorem{remark}{Remark}
\newtheorem{corollary}{Corollary}
\newtheorem{lemma}{Lemma}

\usetikzlibrary{shapes,arrows}
\usetikzlibrary{positioning}

\tikzstyle{block} = [draw, fill=RR_black!0, rectangle, minimum height=0.6cm, minimum width=2cm]
\tikzstyle{rect1} = [rectangle, minimum height=1cm, minimum width=1cm]
\tikzstyle{rect2} = [rectangle, minimum height=1cm, minimum width=0.3cm]
\tikzstyle{rect3} = [fill=RR_black!50,rectangle, minimum height=0.6cm, minimum width=0.6m]
\tikzstyle{comm1} = [draw,dashed, fill=RR_black!10, rectangle, minimum height=1.3cm, minimum width=5.3cm]
\tikzstyle{comm} = [draw,dashed, fill=RR_black!10, rectangle, minimum height=1.6cm, minimum width=5.3cm]
\tikzstyle{rect} = [draw, fill=RR_black!0, rectangle, minimum height=0.6cm, minimum width=0.6cm]
\tikzstyle{sum} = [draw, fill=RR_black!70, circle]
\tikzstyle{input} = [coordinate]
\tikzstyle{output} = [coordinate]
\tikzstyle{pinstyle} = [pin edge={to-,thin,RR_black}]

\definecolor{PU_orange}{RGB}{245,128,37}
\definecolor{PU_orange_light}{RGB}{245,178,78}
\definecolor{RR_blue}{RGB}{16,6,159}
\definecolor{RR_blue_light}{RGB}{0,109,255}
\definecolor{RR_black}{RGB}{30,54,67}
\definecolor{RR_pink}{RGB}{250,70,146}
\definecolor{RR_paragraph}{RGB}{78,93,101}
\definecolor{RR_white}{RGB}{255,255,255}

\IEEEoverridecommandlockouts
\title{\LARGE \bf
	Optimal Control of Systems with Multichannel Packet Loss
	\thanks{ This work is supported by  Rolls-Royce, ESPRC, and The Control, Monitoring and Systems Engineering UTC at The University of Sheffield.}} %Sponsor and financial support acknowledgment goes here. Paper titles should be written in uppercase and lowercase letters, not all uppercase.}
\author{William Casbolt, %
	 Bryn Jones%
	 \thanks{William Casbolt, Bryn Jones, and I\~{n}aki Esnaola are with the Department of Automatic Control and Systems Engineering, University of Sheffield, Sheffield S1 3JD, UK, (e-mail: \{wgcasbolt1, esnaola, b.l.jones\}@sheffield.ac.uk ).}, %
	and I\~{n}aki Esnaola%
	\thanks{{I\~{n}aki Esnaola is also with the Department of  Electrical Engineering, Princeton University, Princeton, NJ 08544, USA.}}% <-this % stops a space
}
\begin{document}

	\maketitle
	\begin{abstract}                % Abstract of not more than 200 words.
	The performance of control systems with input packet losses on the controller to plant communication channel is analysed. %
	%%
	%We prove analytically that a control system communicating optimally with a UDP-like protocol has a higher quadratic cost than a system operating communicating with a TCP-like protocol.
	%%
	The main contribution of this work is a proof that linear optimal control systems operating with UDP-like communication protocols have a larger quadratic cost than the same systems operating with TCP-like protocols.
	%%
	%The communication channel model we consider extends the class of systems modelled in~\cite{1}, such as systems without collocated actuators.
	%%
	The proof is derived for the general case of multidimensional and independent actuation communication channels. %
	In doing so, our results extend previous work to systems with multiple distributed actuators. %
	The difference in cost between two communication protocols is analysed, enabling the maximal difference between the two protocols to be quantified.
	%%
%	The optimal control law for a system operating with communication packet losses is derived under both the TCP-like and the UDP-like communication protocols. %
	%%
	Numerical examples are presented to highlight the difference in costs induced by the choice of communication protocol.
	\end{abstract}
\section{Introduction}
%
%\highlighty{\mbox{Include a motivation section}}
%\begin{itemize}
%	\item distributed control.
%	\item wireless control.
%	\item denial of service attacks.
%	\item distributed denial of service attacks.
%\end{itemize}

%
The surge in number of sensors and smart actuators comes with an increase in wiring within a control system~\cite{LeenG}. %
To this end, wireless communication technology is a readily available tool to reduce the weight and architectural issues of this additional wiring~\cite{KorkuaSurat,AldoThesis,LeenG}. %
Wireless control architectures are attractive in a number of applications, such as aerospace, where reductions in the weight and complexity of the cabling between plant and controller offers significant cost savings~\cite{AldoThesis}, subject, of course, to safety considerations. %
Additionally, to be applicable to a control system with multiple actuators the modelling requires a dedicated communication channel for each actuator as opposed to a single communication channel shared by all actuators. %
However, associated with wireless communication is the issue of packet loss~\cite{KorkuaSurat,AldoThesis}. %
%
%Additionally, packet loss is also a concern with a wired networked system. %
%
%Especially in a network with high traffic. %
%
%Distributed control systems posses actuators that are not necessarily co-located with the controller. %
%
%Therefore, each actuator requires a communication channel to receive the actuation signal sent from the controller. %
%
%Additionally, each actuator may require a dedicated communication channel for a distributed controller. %
%
For communication systems with packet loss, Denial of Service (DoS) attacks are a concern~\cite{ZhangCheng}. %
Moreover, DoS attacks can be implemented within a distributed control system to target individual actuators. %
%
%These attacks require that the communication channel to be generalised to a multichannel in order to account for individual channels for each actuator. %
%
A generalisation of a singular communication channel, as seen in~\cite{1}, is required to analyse the closed-loop performance degradation incurred by packet loss in individual channels. %
The multidimensional communication channel considered is an extension of the actuator channel model used in~\cite{1} where the channel is modelled as a scalar variable. %
The main result is the derivation of the optimal control law for a system communicating over multiple independent lossy actuation communication channels, operating under two different communication protocols. %
This enables the characterisation of the cost difference between the two different communication protocols. % 

In~\cite{1,4,7,8}, two communication protocols are proposed for analysis, namely, a TCP-like protocol and a UDP-like protocol. %
The TCP-like protocol implements an acknowledgement signal that is transmitted back to the controller, confirming whether or not the control signal has been successfully received by the plant. %
In contrast, the UDP-like protocol lacks the acknowledgement link. %
The TCP-like protocol differs from the TCP protocol used in communication literature in that a lost packet is not automatically re-transmitted to the plant since there is no reason for this to be useful any longer to the optimal control input.
This is due to the most recently calculated control signal being the most vital for the plant to receive. %
%
%therefore previous control signals will not be retransmitted if they are dropped. %
%
For example, due to the fact that there is no estimation performed at the plant, in addition to the presence of uncertainty within the system, the most recent measurement contains the least uncertainty of the current state of the system. %
In the event of a packet loss, the plant performs no actuation as is assumed in~\cite{1,4,7,8} and is the main focus of~\cite{9}. %
In~\cite{1,9,7,8,4,12}, control systems with packet loss in the communication channels between the plant and the controller are modelled and analysed. In doing so, the foundations for control and estimation over lossy communication channels are established. %
The optimal control law and estimator is derived for both protocols in \cite{1} using dynamic programming.
%
%derives the optimal control law and the optimal estimator for both protocols.
%
%We focus upon the control problem with packet losses only in the actuator channel and thus assume full state-feedback. 
%
A control system that is susceptible to packet loss on the communication channel from the plant to the controller is considered in~\cite{12}. %
In~\cite{1} and~\cite{8} the work of~\cite{12} is extended to systems with packet loss in both the sensing (plant to controller channel) and the actuation (controller to plant channel) communication channels. %
In~\cite{9}, a comparison of different control strategies in the event of packet loss is studied. %
They study the effect of a smart actuator that supplies the previous input in the event of a packet loss. %
Therein they conclude that there exists a trade off between the zero-input and the hold-input strategies. %
Specifically, in the high packet loss percentage scenario or the `cheap control' scenario the zero-input strategy is superior in terms of Linear Quadratic Gaussian (LQG) cost. %
Where `cheap control' corresponds to a control law that does not penalise the actuators heavily. %
Furthermore, the stability regions for both strategies are the same. %
For that reason, in our work the zero-input strategy is adopted for mathematical simplicity. %
In~\cite{4} and~\cite{7}, systems with and without an acknowledgement link respectively are considered. %
These approaches analyse the performance of the controller, and characterise the trade-off between the control system cost, stability, and the properties of the communication channel.

We consider systems that consists of a plant, a controller and a communication channel, as shown in Fig. \ref{fig:TCP-like-diagram} and~\ref{fig:UDP-like-diagram}. %
This system communicates with one of two communication protocols, a TCP-like protocol or a UDP-like protocol. %
Packet loss only occurs in the controller to actuator communication channel; the sensor to controller communication channel is assumed to be perfect. %
This is done to focus on the impact of the actuation channel within a control setting. %
%
%Additionally, as mentioned in~\cite{1} the seperation between 
%
The analysis in this paper can be extended to the case with a lossy sensor. %
The actuation communication channel is extended from the previous work~\cite{1} to allow for multiple independent channels as opposed to a single channel shared by all actuators. 
As a result of this, the main contribution is an analytical proof that the system cost is always greater as a result of not monitoring realisations of packet loss in the channel. %
The maximal cost difference between the two protocols is also characterised. %
The packet loss in the communication channel is modelled as a set of Independent and Identically Distributed (IID) Bernoulli random variables. %, we term these the packet transmission variables. %
As a result, each actuator either receives the optimal input, or it receives zero input. %
This is equivalent to a probabilistic Denial of Service attack on the communication channel. %
The optimal control law is obtained by formulating the problem in a Model Predictive Control (MPC) framework. %
This subsequently enables the analytical comparison of the cost incurred by both protocols in a more tractable fashion that a dynamic programming approach. %

The structure of the rest of the paper is as follows, Section~\ref{sec:System-model} describes the system model and the MPC framework; %
Section~\ref{sec:MPC-Optimal-Control} contains the main results and the derivation of the optimal control law for both protocols; %
Section~\ref{sec:Cost-Difference-Analysis} characterises the cost difference between the protocols; %
Section~\ref{sec:Numerical-Results} contains numerical results from two illustrative case studies; %
Section~\ref{sec:channel-disc} discusses the implications of utilising a multidimensional channel, based upon the second case study; %
and Section~\ref{sec:Conclusion} presents the conclusion.
\begin{figure}[!t]
	\captionsetup{justification=centering,margin=2cm,width=\linewidth}
	\centering
	\includegraphics[]{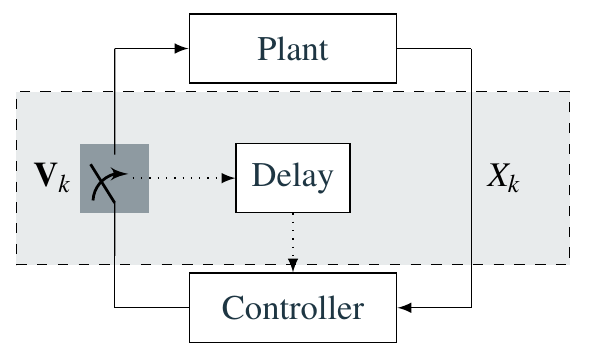}
	\caption[center]{The TCP-like protocol where realisations of the packet transmission variable,~${\bf V}_k$, are transmitted to the controller.}\label{fig:TCP-like-diagram}
\end{figure}
\section{System Model and Problem Formulation}\label{sec:System-model}
\noindent We consider the plant model described by
\beqn\label{eq:plant}
X_{k\splus 1}&=&\Am X_k + \Bm{{\bf V}_k}U_k +W_k, 
\eeqn 
where~$\Am\in \R^{n\times n}$ is the dynamics matrix;~$X_k \in \R^n$ describes the state of the plant at time step~$k \in \N$;~$\Bm \in \R^{n \times m}$ is the control matrix;~${U}_k\in \R^m$ is the vector of control inputs;~$W_k \in \R^n$ is the process noise modelled as a vector of Gaussian random variables with mean~${\bf 0} \in \R^n$ and covariance matrix~$\Sigma_{W} \in S^n_{++}$; where~$S^n_{++}$ is the set of~$n$ by~$n$ symmetric positive definite matrices;~${\bf V}_k \in S^{m}_{+}$ is the packet transmission variable modelled as a diagonal matrix where the~$i$-th diagonal entry is an IID Bernoulli random variable with mean~$\mu_i \in \vR $; and~$S^m_{+}$ is the set of~$m$ by~$m$ symmetric non-negative definite matrices. The initial state of the plant is determined by the Gaussian distributed vector of random variables~$X_k$ with mean~${\overline{X}_k}$ and covariance matrix ~${ \Sigma_{X_k}} \in S^n_{++}$. %
Additionally, the expected value of~$\Vm_k$ is~$\EE[\Vm_k]=\Mm$, where~$\Mm\in S^{m}_{++}$ is a diagonal matrix in which the i-th diagonal element is~$\mu_i$. %
\begin{figure}[!t]
	\captionsetup{justification=centering,margin=2cm,width=\linewidth}
	\centering
	\includegraphics[]{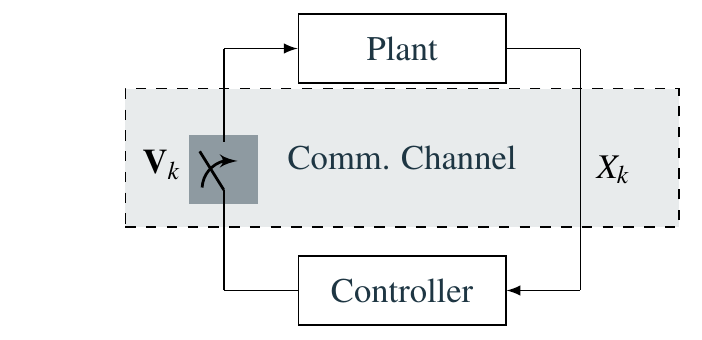}
	\caption{The UDP-like protocol where the realisations of the packet transmission variable,~${\bf V}_k$, are not transmitted to the controller}\label{fig:UDP-like-diagram}
\end{figure}

It is in the structure of~$\Vm_k$ and~$\Mm$ that the system model differs from~\cite{1,9,7,8,4,12}. %
%
%In this paper,~$\Vm_k$ is a diagonal random matrix variable with mean~$\Mm$ as opposed to a scalar random variable with a scalar mean. %
%
Defining~$\Mm$ as a matrix models the case in which the system communicates over~$m$ independent channels, where each actuator has a single dedicated communication channel. %
Specifically, the~$m$-th control input communicates through the~$m$-th channel which is completely characterised by the~$m$-th diagonal entry of~$\Mm$. %
In contrast, were~$\Mm$ and~$\Vm_k$ to be defined as a scalars all control inputs would share a single communication channel that is fully characterised by these scalars. %
This setting generalises the packet loss communication channel seen in~\cite{1}. 
Due to the imperfect communication between the controller and the plant, the operator implements a communication protocol. %
We adopt two protocol paradigms proposed by \cite{1}, namely a UDP-like protocol that does not monitor the communication channel, and a TCP-like protocol that acknowledges receipt of the packet from the controller by sending an {\it acknowledgement} message to the controller over an auxiliary channel. %
As in~\cite{1} it is assumed that the auxiliary channel has perfect communication. %
The difference between both protocol paradigms is depicted in Fig. \ref{fig:TCP-like-diagram} and Fig. \ref{fig:UDP-like-diagram} for the TCP-like and the UDP-like protocols respectively. %
The choice of protocol paradigm for a system results in different information available for the controller. %
We define the information available at the controller for each protocol with the following two information sets
 \newsavebox{\smlVk}% Box to store smallmatrix content
 \savebox{\smlVk}{$\left(\begin{smallmatrix} \Vs^{k-2} &{\bf 0} \\ {\bf 0} & {\bf V}_{k\sminus1} \end{smallmatrix}\right)$}
\beqn
\Ic_k\eqdef
\begin{cases}\label{eq:information}
	\Fc_k=\left\lbrace \Xc^k, \Vc^{k\sminus1}\right\rbrace,&\mbox{TCP-like},\\
	\Gc_k=\left\lbrace \Xc^k\right\rbrace,&\mbox{UDP-like},\\
\end{cases}
\eeqn
%
%	 \beqn\label{eq:information}
%\Ic_k = \left\{ \ba{cc} \Fc_k=\left\{X_k, \Vs^{k\sminus1} \right\},& \mbox{TCP-like} \\ \Gc_k=\left\{X_k\right\},\qquad \ \ \ & \mbox{UDP-like}\hfill \ea\right. 
%\eeqn
%
\noindent where~$\Vc^{k\sminus1}\eqdef \left\lbrace \Vm_0, \Vm_1, \ldots, \Vm_{k\sminus1}\right\rbrace$ and~$\Xc^k\eqdef \left\{X_0,X_{1},\dots,X_k\right\}$. %
Note that all sets are monotonically increasing, i.e.~$\Ic_{k} \subseteq \Ic_{k\splus 1}$. % 
Additionally, the lower indices represent the time index whereas upper indices refers to the dimension. %
The information set at each time step contains the information from all previous time steps in addition to the information from the current time step. %
Under the TCP-like protocol the controller has access to the realisation of the packet transmission variable,~${\bf V}_k$, when performing state estimation and incorporates it in the error prediction to obtain an estimate with error
\begin{subequations}\label{eq:err-scalar}
	\beqn\label{eq:TCP-scalar-err}
	E_{k\splus 1}\! \left(\! \Fc_k \! \right)\nquad \ &\eqdef&\nquad \  X_{k\splus 1}-\EE\left[X_{k\splus 1}{\Big |}\Fc_k,{\bf V}_k\right]\nonumber \\
	&=& \nquad \  \Am  X_k +  \Bm{{\bf V}_k}U_k \! \left(\! \Fc_k \! \right) +W_k - \Am  \widehat{X}_k -  \Bm{{\bf V}_k}U_k \! \left(\! \Fc_k \! \right) \nonumber\\
	&=& \nquad \  \Am E_{k}\!\!\left(\!\Fc_{k\sminus1}\!\right)\!\! +W_k,
	\eeqn 
where~$\widehat{X}_{k} \eqdef \EE\left[X_k\right]$. %
%
%Note that the random variable~$U_k$, is redefined as the function~$U_k\!\left(\! \Ic_k \! \right)$. %
%
The function~$U_k\!\left(\! \Ic_k \! \right)$ is a function of the information set~$\Ic_k$ and takes the form of a state feedback law, and is therefore, a random variable,. %
The UDP-like protocol error prediction differs from the TCP-like protocol in that there is no knowledge of the realisation of~$\Vm_k$, and therefore, the error for the UDP-like protocol is given by
	\beqn
	E_{k\splus 1}\! \left(\! \Gc_k \! \right) \nquad \ &\eqdef& \nquad \  X_{k\splus 1}-\EE\left[X_{k\splus 1}{\Big |}\Gc_k\right] \nonumber \\
	&=& \nquad \   \Am  X_k +  \Bm{{\bf V}_k}U_k \! \left(\! \Gc_k \! \right) +W_k - \Am  \widehat{X}_k -  \Bm\Mm U_k \! \left(\! \Gc_k \! \right) \nonumber\\
	&=& \nquad \  \Am E_{k}\!\!\left(\! \Gc_{k\sminus1} \!\right)\!\! + \Bm\left({\bf V}_k - \Mm\right)U_k \! \left(\! \Gc_k \! \right) + W_k.
	\eeqn 
\end{subequations}
Note that the error prediction for both the UDP-like and the TCP-like protocol resemble those in~\cite{1}. %
%
%The difference is that~$\Vm_k$ in our setting is a diagonal matrix rather than a scalar. It is also interesting to note that for the UDP-like protocol, as in~\cite{1}, the error term at time step~$k+1$ depends on~$U_k$. %
%
As shown in~\cite{1}, the optimal linear control law for the UDP-like protocol can only be obtained when perfect state information is available. %
Indeed, the lack of knowledge about the packet loss breaks the separation structure between optimal estimation and optimal control. %
%
%It is also shown in~\cite{1} that the UDP-like protocol requires full state observation for optimality. %
%
%Additionally, due to the fact that the loss variable is a matrix in our formulation the problem becomes more amenable to be cast in a Model Predictive Control (MPC) framework with the assumption of full state information. %
%
It is here our derivation diverges again from~\cite{1}. %
Assuming perfect knowledge of the realisation of~$X_k$,~(\ref{eq:plant}) is expanded over a time horizon~$N\in\N$ as follows:
	 \beqn\label{eq:predict}\!\!\!\!
	\ba{c}
	X_{k\splus 1} \\ X_{k\splus 2} \\ \vdots \\ X_{k\splus N}
	\ea\!\!\!\!\!\!\!\!
	\ba{c}
	=\\=\\ {\vdots }\\=
	\ea 
	\nquad
	\ba{c}
	\Am  X_k \\\Am  X_{k\splus 1} \\ \vdots \\ \Am  X_{k\splus N\sminus1}
	\ea
	\nquad
	\ba{c}
	+\\+\\\vphantom{\vdots } \\+
	\ea 
	\nquad
	\ba{c}
	\Bm{\bf V}_kU_k \! \left(\! \Ic_k \! \right) \\ \Bm{\bf V}_{k\splus 1} U_{k\splus 1}\! \left(\! \Ic_k \! \right) \\ \vdots\\  \Bm{\bf V}_{k\splus N\sminus1}U_{k\splus N\sminus1}\! \left(\! \Ic_k \! \right)
	\ea
	\nquad
	\ba{c}
	+\\+\\\vphantom{\vdots } \\+
	\ea 
	\nquad
	\ba{c}
	W_k ,\\
	W_{k\splus 1} ,\\
	\vdots\\
	W_{k\splus N\sminus1}.
	\ea\!\!\!\!\!\!
	\eeqn 
Exploiting the recursive structure yields
	 \beqn\label{eq:substitution-predict}
	 X_{k\splus1} \nquad   &=& \nquad  \Am  X_k +	\Bm {\bf V}_k U_k \! \left(\! \Ic_k \! \right) +	W_k,\nonumber \\
	 X_{k\splus2} \nquad  &=& \nquad  \Am^2 X_k +\Am\Bm{\bf V}_k U_k \! \!\left(\! \Ic_k \! \right) + \Bm{\bf V}_{k\splus1} U_{k\splus1}\! \! \left(\! \Ic_k \! \right)% \nonumber \\
	 %
%	 && \qquad
 	 +	\Am W_k +W_{k\splus1},\nonumber\\
	 &\vdots& \nonumber \\
	 X_{k\splus N} \nquad  &=& \nquad \Am^N  X_k\! +\! \Am^{N\sminus1}\Bm {\bf V}_{k}U_k \! \left(\! \Ic_k \! \right)\! +\! \dots \nonumber \\
	 && \nquad  +  \Bm {\bf V}_{k\splus N\sminus1} U_{k\splus N\sminus1}\! \left(\! \Ic_k \! \right) + \Am^{N\sminus1}W_k + \dots + W_{k\splus N\sminus1}.\nonumber
	\eeqn 
\begin{figure*}[!t]
	\fontsize{9pt}{2pt}
	\begin{align}\label{big}\!\!\!\!
		\underset{\pazocal{X}_k}{\underbrace{\left(\!\!\!\!\ba{c}X_{k\splus 1} \\ X_{k\splus 2} \\ \vdots \\ X_{k\splus N}
				\ea\!\!\!\!\right)}}\!\!\!
		= \!\!\!
		\underset{\Phi}{\underbrace{\left(\!\!\!\!\ba{c}
				\Am \\ \Am^2 \\ \vdots \\ \Am^N 
				\ea\!\!\!\!\right)}}\!X_k\!
		+ \!\!\!
		\underset{\Gamma}{\underbrace{\left(\!\!\!\!\ba{cccc}
				\Bm &\bf 0 &\dots & \bf 0\\
				\Am\Bm & \Bm& \ddots& \vdots\\
				\vdots& \ddots& \ddots& \bf 0 \\
				\Am^{N\sminus1}\Bm & \dots & \Am \Bm & \Bm
				\ea\!\!\!\!\right)}}\
		\underset{\Vs}{\underbrace{\left(\!\!\!\!\!\!\ba{cccc}
				{\bf V}_k &\bf 0 &\dots & \bf 0 \\
				\bf 0 & {\bf V}_{k\splus 1} & \ddots & \vdots \\
				\vdots & \ddots & \ddots & \bf 0 \\
				\bf 0 &\dots & \bf 0 & {\bf V}_{k\splus N\sminus1}
				\ea\!\!\!\!\!\! \right)}}\!\!
		\underset{\pazocal{U}_k \! \left(\! \Ic_k \! \right)}{\underbrace{\left(\!\!\!\!\!\!\ba{c}
				U_k \!\! \left(\! \Ic_k \! \right)\\
				U_{k\splus 1}\!\! \left(\! \Ic_k \! \right)\\
				\vdots\\
				U_{k\splus N\sminus1}\!\! \left(\! \Ic_k \! \right)
				\ea\!\!\!\!\!\! \right)}}\!\!
		+\!\!
		\underset{\Lambda}{\underbrace{\left(\!\!\!\!\!\!\ba{cccc}
				{\bf I} &\bf 0 &\dots & \bf 0\\
				\Am & {\bf I} & \ddots& \vdots\\
				\vdots& \ddots& \ddots& \bf 0 \\
				\Am^{N\sminus1} & \dots & \Am & {\bf I}
				\ea\!\!\!\!\right)}}\!
		\underset{\pazocal{W}_k}{\underbrace{\left(\!\!\!\!\!\!\ba{c}
				W_k \\
				W_{k\splus 1}\\
				\vdots \\
				W_{k\splus N\sminus1}
				\ea\!\!\!\!\!\!\right)}}
	\end{align}
	\hrulefill
\end{figure*}
\normalsize
$\!\!$Re-writing (\ref{eq:substitution-predict}) in matrix form yields (\ref{big}) seen at the top of the next page. %
Re-casting (\ref{big}) as a prediction matrix equation gives
\beqn\label{eq:matrix-plant}
{\pazocal{X}}_k\eqdef\Phi X_k +\Gamma {\Vs} \pazocal{U}_k \! \left(\! \Ic_k \! \right) + \Lambda {\pazocal{W}}_k,
\eeqn 
where~$\Phi \in \R^{Nn\times n}$ is the dynamics matrix over the prediction horizon;~%
$\pazocal{X}_k \in \R^{Nn}$ is the state prediction vector;~%
$\Gamma \in \R^{Nn\times Nm}$  is the propagation matrix for the control over the prediction horizon;~%
${\pazocal{U}_k}\!\!\left(\!\Ic_k\!\right)\!\! \in \R^{Nm}$ is the realisation at time step~$k$ of the control law computed with access to the information set $\Ic_k$;~%
$\Lambda \in \R^{Nn\times Nn}$ is the propagation matrix for the process noise;~%
${\pazocal{W}_k} \in \R^{Nn}$ is the process noise over the prediction horizon with mean~$\bf 0$ and covariance~$\Sigma_{\pazocal{W}}$;~%
$\Sigma_{\pazocal{W}}\in S^{Nn}_{++}$ is the diagonal block matrix where the~$i$-th block is~$\Sigma_{W}$;~
$\Vs\in S^{Nm}_{+}$ is a diagonal matrix with the Bernoulli random variables describing the packet transmission over the prediction horizon in the diagonal;%
and~$\Vsb\in S^{Nm}_{++}$ is the block diagonal matrix where the~$i$-th block is~$\Mm$, and therefore,~$\EE\left[\Vs \right]=\Vsb$. %
%
%\highlightr{\mbox{Add a sentence saying~$\Gamma=\Lambda({\bf I}\otimes \Bm)$??}}
%
To control the system over the horizon, N, the controller calculates the expected state trajectory,~$\widehat{\pazocal{X}}_k$. %
Note that for both protocols the estimate coincides, due to the fact that neither protocol knows the realisation of~${\bf V}_k$ before actuating. %
The expected state trajectory for both protocols is therefore given by
	 \beqn\label{eq:pred-err}
	\widehat{\pazocal{X}}_{k}^{} \eqdef \EE\left[{\pazocal{X}}_k{\Big |}\Ic_k\right] = \Phi X_k +\Gamma \Vsb{\pazocal{U}_k \! \left(\! \Ic_k \! \right)} .
	\eeqn 
%%	
%It should be noted that for estimation, the TCP-like protocol has access to the previous packet realisations, which results in (\ref{eq:TCP-scalar-err}). However, when computing the optimal control law the system operator does not have access to future packet transmission realisations, which results in (\ref{eq8}). This preserves causality, since the operator does not know if a packet is lost before transmission. 
%%
%
In the TCP-like regime the operator does not know the realisation of a packet transmission before actuating, which results in~(\ref{eq:pred-err}), but knows the packet transmission realisation when updating the state estimate, which results in (\ref{eq:TCP-scalar-err}). %
The TCP-like protocol only estimates the packet transmission for the optimal control problem. %
In contrast, the UDP-like protocol packet transmission variables are estimated for both the estimation and the optimal control problem. %
Expanding the update error terms in~(\ref{eq:err-scalar}) over the prediction horizon of~$N$ time-steps results in
	\begin{subequations}
	 \beqn
	\Em_k \! \left(\! \Fc_k \! \right) \nquad \ &\eqdef& \nquad \ {\pazocal{X}}_k - \EE\left[{\pazocal{X}}_k{\Big |}\Fc_k, \Vs \right] \nonumber \\
	&=&\nquad \ \Phi X_k +\Gamma {\Vs}{\pazocal{U}_k \! \left(\! \Fc_k \! \right)} + \Lambda {\pazocal{W}}_k  - \widehat{\pazocal{X}}_{k} \nonumber \\
	&=&\nquad \  \Phi X_k +\Gamma {\Vs}{\pazocal{U}_k \! \left(\! \Fc_k \! \right)} + \Lambda {\pazocal{W}}_k - \Phi X_k -\Gamma {\Vs}{\pazocal{U}_k \! \left(\! \Fc_k \! \right)} \nonumber\\
	&=&\nquad \   \Lambda {\pazocal{W}}_k, \label{eq:TCP-prediction-error}
	\eeqn
	\beqn
	\Em_k \! \left(\! \Gc_k \! \right) \nquad \ &\eqdef& \nquad \ {\pazocal{X}}_k - \EE\left[{\pazocal{X}}_k{\Big |}\Gc_k\right] \nonumber \\
	&=& \nquad \ \Phi X_k +\Gamma {\Vs}{\pazocal{U}_k \! \left(\! \Fc_k \! \right)} + \Lambda {\pazocal{W}}_k  - \widehat{\pazocal{X}}_{k} \nonumber \\
    &=& \nquad \ \Phi X_k +\Gamma {\Vs}{\pazocal{U}_k \! \left(\! \Gc_k \! \right)} + \Lambda {\pazocal{W}}_k - \Phi X_k -\Gamma \Vsb{\pazocal{U}_k \! \left(\! \Gc_k \! \right)} \nonumber\\
	&=&  \Gamma\left(\Vs - \Vsb\right){\pazocal{U}_k \! \left(\! \Gc_k \! \right)} + \Lambda {\pazocal{W}}_k . \label{eq:UDP-prediction-error}
	\eeqn 	
	\end{subequations}
In this setting we formulate a Linear Quadratic Gaussian (LQG) control problem, i.e. the system operator minimises a quadratic function of the states and inputs. %
This function is weighted with diagonal state penalty matrix~$\Omega\in S_{++}^{Nn}$, diagonal input penalty matrix~$\Psi\in  S_{++}^{Nm}$, and diagonal matrix~$\Qm\in S^n_{++}$. %
Note that the penalties at each time step may vary. %
Since~$\pazocal{W}_k$ is random, then the state of the plant is random, which yields a stochastic model predictive control problem~\cite{5}. %
The cost function to be minimised is the expected cost defined as
\begin{subequations}
	 \beqn
	 %J(X_k,{\pazocal{U}_k},\Vs)&\eqdef& X_k^{\sf T} \Qm X_k + {\pazocal{X}}_k^{ {\sf T}} \Omega {\pazocal{X}}_k + \pazocal{U}^{\sf T}_k \Vs^{\!\!\!\! \sf T} \Psi \Vs {\pazocal{U}_k},\nonumber\\
	 %
J\left(\! \Ic_k \! \right)\nquad \ &\eqdef& \nquad \ \EE\left[\left. X_k^{\sf T} \Qm X_k + {\pazocal{X}}_k^{\sf T} \Omega {\pazocal{X}}_k \right.\right. \nonumber \\
&& \qquad \qquad \left.\left. + \pazocal{U}^{\sf T}_k\!\! \left(\! \Ic_k \! \right)\!\! \Vs^{\!\!\!\! \sf T} \Psi \Vs {\pazocal{U}_k}\!\! \left(\! \Ic_k \! \right)\! \right|\Ic_k\right]\!\!,\!\! \label{eq:plant0}
	\eeqn 
\noindent where the expectation in (\ref{eq:plant0}) is with respect to the joint distribution of~$\Vs$ and~${\pazocal{W}}_k$. %
The expectation is taken sequentially as in \cite[Lemma~$1$(c)]{4} to account for the causality constraints imposed by the system. %
Therein, the expectation at each time step is conditioned on all previous time steps. %
This is due to the fact that the sequence of states at each time step forms a Markov chain, i.e.~$X_k\rightarrow X_{k\splus 1}\rightarrow\dots\rightarrow X_{k\splus N}$. %
The state trajectory~${\pazocal{X}}_k$ is re-written in terms of the estimate~$\widehat{\pazocal{X}}_{k}^{}~$ and the error induced by the estimate~$\Em_{k}$. %
Substituting~$\pazocal{X}_k =\widehat{\pazocal{X}}_k+\Em_{k}$ into (\ref{eq:plant0}) yields
	 \beqn\label{eq:lqg-sub-err}
	\!\!J\!\! \left(\! \Ic_k \! \right)\nquad \ &=& \nquad \ \EE\!\! \left[\left.\! X_k^{\sf T}\!\! \Qm X_k\! + \! \left(\widehat{\pazocal{X}}_{k}^{}\!\! + \! {\Em_{k}}\right)^{\sf\!\! T}\!\!\! \Omega \! \left(\widehat{\pazocal{X}}_{k}^{}\!\! + \! {\Em_{k}}\right)\! \right. \right. \nonumber \\
	&& \qquad \qquad  \left. \left. \vphantom{\left(\widehat{\pazocal{X}}_{k}^{}\!\! + \! {\Em_{k}}\right)^{\sf\!\! T}} %
	+ \! {\pazocal{U}_k}^{\sf\!\!\! T}\!\! \left(\! \Ic_k \! \right)\!\! \Vs^{\sf\!\! T} \! \Psi \Vs {\pazocal{U}_k}\!\! \left(\! \Ic_k \! \right)\! \right|\! \Ic_k \!\! \right]\!\! .
	\eeqn 
	\end{subequations}
The optimal control problem is to find the input sequence~$\pazocal{U}^{*}_{k}$ that minimises (\ref{eq:lqg-sub-err}). %
Additionally, it should be noted that~$\EE\left[\Em_{k}{\big |}\Ic_k\right]=0$ and the state error and the state estimate are independent for both protocols. %
The proofs of these statements are provided in~\cite{4}. %
which leads to the following optimal cost definition:
\beqn\label{eq:opt-lqg}
&&\nqquad J^*\left(\! \Ic_k \! \right)\eqdef \min_{{\pazocal{U}_k}\!\! \left(\! \Ic_k \! \right)\!}\left\{\EE \left[X_k^{\sf T} \Qm X_k+ \widehat{\pazocal{X}}_{k}^{\sf T}\Omega\widehat{\pazocal{X}}_{k}^{}+ \Em_{k}^{\sf T} \Omega{\Em_{k}}\vphantom{\Big|} \right.\right. \nonumber \\
	&& \qquad \qquad \quad  \left. \left. \vphantom{\left(\widehat{\pazocal{X}}_{k}^{}\!\! + \! {\Em_{k}}\right)^{\sf\!\! T}} %
+ \! {\pazocal{U}_k}^{\sf\!\!\! T}\!\! \left(\! \Ic_k \! \right)\!\! \Vs^{\sf\!\! T} \! \Psi \Vs {\pazocal{U}_k}\!\! \left(\! \Ic_k \! \right)\! \right|\! \Ic_k \!\! \right]\!\! .
	\eeqn 
\section{MPC Optimal Cost Derivation and Analysis}\label{sec:MPC-Optimal-Control}
The derivation of the optimal control law is therefore recast into solving the minimisation in (\ref{eq:opt-lqg}) for both communication protocols. %
The first and second term on the right hand side of (\ref{eq:opt-lqg}) are not random due to the information available and are therefore unaffected by the expectation. %
Furthermore, since the first term does not depend on~${\pazocal{U}_k}$ the minimisation is rewritten as
	 \beqn\label{eq:opt-lqg-simp}
	J^*\left(\! \Ic_k \! \right)\!\!\!\!\!\!\!\!&=&\!\!\!\!\!\!\!\! X_k^{\sf T} \Qm X_k+ \min_{{\pazocal{U}_k}\!\! \left(\! \Ic_k \! \right)\!}\left\{ \widehat{\pazocal{X}}^{\sf T}_{k}\Omega\widehat{\pazocal{X}}_{k}^{} \right. \nonumber \\
	&& \left. \!\! + \EE \left[\!\! \left. \vphantom{\widehat{\pazocal{X}}^{\sf T}_{k} } %
	 \Em_{k}^{\sf T} \Omega{\Em_{k}} + \! {\pazocal{U}_k}^{\sf\!\!\! T}\!\! \left(\! \Ic_k \! \right)\!\! \Vs^{\sf\!\! T} \! \Psi \Vs {\pazocal{U}_k}\!\! \left(\! \Ic_k \! \right)\! \right| \Ic_k\right]\right\}\!\!. \quad \ 
	\eeqn 
%	
%The quadratic structure of (\ref{eq:opt-lqg-simp}) implies that computing the expectation is involved. %
%
The computation of the expectation of the last term can be simplified by using the commutation properties of diagonal matrices and the idempotency of the matrix~$\Vs$. %
Additionally, due to the causality imposed on the system,~${\pazocal{U}_k}\!\! \left(\! \Ic_k \! \right)\!$ does not depend on the future {\it realisations} of~$\Vm_k$ or~$W_k$, and therefore, is not affected by the expectation. %
Note, that this still allows for a~${\pazocal{U}_k}\!\! \left(\! \Ic_k \! \right)\!$ that depends on the statistics of each of these variables, just not the future {\it realisations}. %
The last term in (\ref{eq:opt-lqg-simp}) is %
	\beqn
	\EE\left[ \left. \! {\pazocal{U}_k}^{\sf\!\!\! T}\!\! \left(\! \Ic_k \! \right)\!\! \Vs^{\sf\!\! T} \! \Psi \Vs {\pazocal{U}_k}\!\! \left(\! \Ic_k \! \right)\! \right|\Ic_k\right] \nquad \ &=& \nquad \ 
	%\EE\left[ \pazocal{U}^{\sf T}_k \Vs^{\!\!\!\! \sf T}\Vs \Psi {\pazocal{U}_k}{\Big |}\Ic_k\right]\nonumber\\
	%
%	&=& \EE\left[\pazocal{U}^{\sf T}_k \Vs \Psi{\pazocal{U}_k} {\Big |}\Ic_k\right],\nonumber\\
	%
	{\pazocal{U}_k \! \left(\! \Ic_k \! \right)^{\!\! \sf T}} \Vsb \Psi {\pazocal{U}_k \! \left(\! \Ic_k \! \right)}. \ \ \ \label{eq:plant4}
	\eeqn 
Therefore, (\ref{eq:opt-lqg-simp}) is equivalent to %
	 \beqn\label{eq:LQG-cost}
	&&\nqquad J^*\left(\! \Ic_k \! \right)= X_k^{\sf T} \Qm X_k + \min_{{\pazocal{U}_k}\!\! \left(\! \Ic_k \! \right)\!}\left\{ \widehat{\pazocal{X}}_{k}^{\sf T}\Omega\widehat{\pazocal{X}}_{k}^{} + {{\pazocal{U}_k \! \left(\! \Ic_k \! \right)}{^{\sf T}}} \Vsb \Psi {\pazocal{U}_k \! \left(\! \Ic_k \! \right)}\vphantom{\Big |}\right. \nonumber \\
	&& \qquad\qquad\qquad\qquad\qquad\left.  +\EE \left[ \Em_{k}^{\sf T} \Omega{\Em_{k}}{\Big |}\Ic_k\right]\right\}.
	\eeqn 
The term involving the expected state trajectory is combined with~(\ref{eq:pred-err}) to give
	\beqn
	J^*\left(\! \Ic_k \! \right) \nquad \  &=& \nquad \  X_k^{\sf T} \left(\Qm + \Omega_{p}\right)X_k
%	X_k^{\sf T} \Qm X_k \nonumber \\
%	&&\qquad+ \min_{{\pazocal{U}_k}\!\! \left(\! \Ic_k \! \right)\!}\left\{ \left(\Phi X_k +\Gamma \Vsb{\pazocal{U}_k \! \left(\! \Ic_k \! \right)}\right)^{\sf T}\Omega\left(\Phi X_k +\Gamma \Vsb{\pazocal{U}_k \! \left(\! \Ic_k \! \right)}\right) \right. \nonumber \\
%	&& \qquad\left. + {\pazocal{U}_k \! \left(\! \Ic_k \! \right)^{\!\! \sf T}} \Vsb \Psi {\pazocal{U}_k \! \left(\! \Ic_k \! \right)} +\EE \left[ \Em_{k}^{\sf T} \Omega{\Em_{k}}{\Big |}\Ic_k\right]\right\} , \nonumber	\eeqn
%	\beqn	
 + \min_{{\pazocal{U}_k}\!\! \left(\! \Ic_k \! \right)\!}\left\{\EE \left[ \Em_{k}^{\sf T} \Omega{\Em_{k}}{\Big |}\Ic_k\right] \right. \nonumber \\
	&&  \nquad  \left. + {\pazocal{U}_k \! \left(\! \Ic_k \! \right)^{\!\! \sf T}} \Vsb \left(2\Omega_{gp} X_k + \left(\Omega_{g} \Vsb+ \Psi \right){\pazocal{U}_k \! \left(\! \Ic_k \! \right)}\right)\!\! \right\}\!\! , \ \ \ \label{eq:opt-lqg-simplest}
	\eeqn 
\noindent where~$\Omega_{p}=\Phi^{\sf T}\Omega\Phi$,~$\Omega_{g}=\Gamma^{\sf T}\Omega\Gamma$, and~$\Omega_{gp}=\Gamma^{\sf T}\Omega\Phi$. 

Evaluating the quadratic error requires knowledge of second order statistics. %
It is in this step that the differences between the UDP-like protocol and the TCP-like protocol become apparent. %
This observation leads to the first lemma.
\begin{lemma}\label{lem:error-proof}
Consider the system modelled by (\ref{eq:plant}) with access to (\ref{eq:information}). %
Then the following holds
\begin{subequations}
	 \beqn
\!\!  \EE\!\! \left[\! \Em_{k}^{\sf T}\! \Omega{\Em_{k}} {\Big |}\Fc_k\right]\nquad \! &=&\nquad  \tr\left( \! \Omega_{l} \Sigma_{\pazocal{W}} \! \right)\!\!,\label{eq:TCP-err}\\
\!\! \EE\!\! \left[ \Em_{k}^{\sf T} \Omega{\Em_{k}} {\Big |}\Gc_k\right]\nquad  &=&\nquad  {\pazocal{U}_k \! \left(\! \Gc_k \! \right)^{\!\! \sf T}}\!  \Vsb\! \left({\bf I}\odot \Omega_{g}\right) \! \left({\bf I}- \Vsb\right) \!\! {\pazocal{U}_k \! \left(\! \Gc_k \! \right)} \! \nonumber \\
&& \qquad \qquad  + \! \tr\left( \! \Omega_{l} \Sigma_{\pazocal{W}} \! \right)\!\!,\ \quad  \label{eq:UDP-err}
	\eeqn 
\end{subequations}
where~$\Omega_{l}=\Lambda\Omega\Lambda$ and~$\odot$ denotes the Hadamard product.
\end{lemma}
\proof
See Appendix.
\newsavebox{\smlA}% Box to store smallmatrix content
\savebox{\smlA}{$\left(\begin{smallmatrix} 1.03&0.005\\0.35&0.5\end{smallmatrix}\right)$}
\newsavebox{\smlV}% Box to store smallmatrix content
\savebox{\smlV}{$\left(\begin{smallmatrix}0.9&0\\0&0.5\end{smallmatrix}\right)$}

Lemma \ref{lem:error-proof} highlights that the UDP-like quadratic error term depends on~$\pazocal{U}_k \! \left(\! \Gc_k \! \right)$, whereas the TCP-like protocol does not. %
Additionally,~(\ref{eq:opt-lqg-simplest}) shows that the quadratic error term lies within the minimisation. %
Therefore, the term for the TCP-like quadratic error is removed from the minimisation in~(\ref{eq:opt-lqg-simplest}) whereas the UDP-like term is not. %
Due to this, the derivation of the optimal control law is at this point split into two cases.
\begin{theorem}\label{th:contol}
	Consider the closed-loop systems shown in Fig.~\ref{fig:TCP-like-diagram} and Fig.~\ref{fig:UDP-like-diagram}, with plant dynamics given in~(\ref{eq:plant}), protocol dependent information sets given in~(\ref{eq:information}), and controller cost function given in~(\ref{eq:opt-lqg}), respectively. %
	Then the optimal cost for the TCP-like protocol is
%
%	Consider a Gauss-Markov system described by~(\ref{eq:plant}) that is experiencing actuation packet losses. The system operator implements either a TCP-like or UDP-like protocol, depicted in Fig. \ref{fig:UDP-like-diagram} and Fig. \ref{fig:TCP-like-diagram}. The optimal cost for the TCP-like protocol is
	%
	\begin{subequations}
	
	\beqn\label{eq:TCP-opt-cost}
	J^*\left(\! \Fc_k \! \right)\nquad \ &=& \nquad \ \!\! X_k^{\sf T}\!\!\left( \Qm +\Omega_{p} \right)\!\!X_k\!+\! \tr\left(\Sigma_{\pazocal{W}}\Omega_{l}\!\right)\!\! \nonumber \\
	&& \qquad\qquad  -\!  X_k^{\sf T}\! \Omega_{gp}^{\sf T}\!\Gm^{\sminus1}\!\!\left(\!\Fc_k\!\right)\!\Vsb \Omega_{gp} X_k ,
	\eeqn
	and the optimal cost for the UDP-like protocol is
	\beqn\label{eq:UDP-opt-cost}
	J^*\left(\! \Gc_k \! \right) \nquad \ &=&\nquad \ \!\! X_k^{\sf T}\!\!\left( \Qm +\Omega_{p} \right)\!\!X_k\!+\! \tr\left(\Sigma_{\pazocal{W}}\Omega_{l}\!\right)\!\! \nonumber \\
	&& \qquad\qquad  -\!  X_k^{\sf T}\! \Omega_{gp}^{\sf T}\!\Gm^{\sminus1}\!\!\left(\!\Gc_k\!\right)\!\Vsb \Omega_{gp} X_k.
	\eeqn		
\end{subequations}
	The corresponding optimal control laws are
	\begin{subequations}\label{eq:control-laws}
	\beqn
	{\pazocal{U}^{^*}_k \! \left(\! \Fc_k \! \right)}\nquad \  &\eqdef&\nquad \ - \left(\Omega_{g} \Vsb +\Psi\right)^{\sminus1} \Omega_{gp} X_k \label{eq:opt-law-1}\\
	&=& \nquad \ - \Gm^{\sminus1}\!\!\left(\!\Fc_k\!\right)\Omega_{gp} X_k, \nonumber \\
	{\pazocal{U}^{^*}_k \! \left(\! \Gc_k \! \right)} \nquad \  &\eqdef&\nquad \ - \left(\Psi +\left({\bf I} \odot \Omega_{g}\right)\left({\bf I} -\Vsb\right) + \Omega_{g}\Vsb\right)^{\sminus1} \Omega_{gp} X_k \quad \  \label{eq:opt-law-2} \\
	&=& \nquad \ - \Gm^{\sminus1}\!\!\left(\!\Gc_k\!\right)\Omega_{gp} X_k, \nonumber 
	\eeqn	
\end{subequations}
	for the TCP-like and the UDP-like protocols, respectively.
\end{theorem}

\begin{proof}
 See Appendix. %
\end{proof}

\begin{remark}%\highlightr{\mbox{insert remark}}
%Note the differences between the two protocols and their corresponding optimal control laws. %
%
In the TCP-like regime the optimal control law,~(\ref{eq:opt-law-1}), only depends on the mean number of packet transmissions,~$\Vsb$, and this term weights how the actuation propagates through the system via the~$\Omega_{g}$ term. %
On the other hand, the optimal control law of the UDP-like regime,~(\ref{eq:opt-law-2}), contains an additional term that weighs the control law with the probability of packet loss,~${\bf I} - \Vsb$. %
%
%Note that as a result of this additional term, the UDP-like protocol weights the control law with a function of the variance of~$\Vs$. %
%
%The differing information about~$\Vs$ causes the error to depend on its second order statistics. %
\end{remark}

The optimal control laws presented in~(\ref{eq:opt-law-1}) and~(\ref{eq:opt-law-2}) and the corresponding optimal cost functions~(\ref{eq:TCP-opt-cost}) and~(\ref{eq:UDP-opt-cost}) depend on~$\Vsb$. %
Therefore, the current formulation makes no assumption on the stationarity of the random process governing the channel-loss statistics. %
Specifically, much like how the penalty matrices,~$\Psi$ and~$\Omega$, vary along the time horizon, the mean of packet transmission for each channel may also vary over the time horizon. %
This allows for a wider class of packet loss models to be utilised. %
For example, a sequence of packet losses that form a Markov chain. %
In this scenario the expected value of a packet transmission,~$\Vm_k$, is modelled as~$\Vm_k \sim \Bc \left(\Mm_k\right) $ where~$\Mm_k\in S^{m}_{++}$ is a diagonal matrix in which the i-th diagonal element is~$\mu_{i,k}$ which describes the probability of a packet transmission in the~$i$-th channel at the~$k$-th time step. %
Therefore,~$\EE[\Vm_k]=\Mm_k$ and~$\EE\left[\Vs\right] = \Vsb$ where~$\Vsb$ is the block diagonal matrix where the~$i$-th block is~$\Mm_k$. %
Substitution of these definitions into the above derivation does not break any assumptions made and results in a control law and optimal cost function for a non-stationary sequence of packet losses. %

\section{Cost Difference Analysis}\label{sec:Cost-Difference-Analysis}

The difference in the information sets leads to different optimal control laws, as seen in~(\ref{eq:opt-law-2}), and results in differing costs over the horizon. %
In the following it is shown that the expected optimal control cost incurred by the information set of the UDP-like protocol is strictly greater than the expected optimal control cost incurred by using the information set of the TCP-like protocol. %
% for a communication channel not corresponding to perfect communication, i.e.~${\bf 0}\prec\Mm\prec{\bf I}$.
%
\begin{theorem}[Main Result]\label{th:analytic-cost-increase}
	Let~$\Mm$ such that~${\bf 0}\prec\Mm\prec {\bf I}$, with information sets given in~(\ref{eq:information}). Then
	\beqn
J^*\left(\! \Gc_k \! \right)-J^*\left(\! \Fc_k \! \right) > 0,\nonumber 
\eeqn
where the optimal cost is as defined in~(\ref{eq:opt-lqg})
%
%\beqn
%J^*\left(\! \Ic_k \! \right)\nquad \ &\eqdef& \nquad \ \min_{{\pazocal{U}_k}\!\! \left(\! \Ic_k \! \right)\!}\left\{ \!\!\EE\left[\left. X_k^{\sf T} \Qm X_k + {\pazocal{X}}_k^{\sf T} \Omega {\pazocal{X}}_k \right.\right.\right.\nonumber \\
%&& \qquad \qquad \left.\left.\left. + \pazocal{U}^{\sf T}_k\!\! \left(\! \Ic_k \! \right)\!\! \Vs^{\!\!\!\! \sf T} \Psi \Vs {\pazocal{U}_k}\!\! \left(\! \Ic_k \! \right)\! \right|\Ic_k\right]\!\!\right\}\!\!.\!\!\nonumber
%\eeqn

%	Assume that~$\Mm$ is diagonal and~${\bf 0}\prec\Mm\preceq {\bf I}$. For a control system with actuation packet losses as modelled in (\ref{eq:plant}) using the cost function (\ref{eq:LQG-cost}), the expected cost with a UDP-like protocol is greater than the TCP-like cost, with equality only achieved at~$\Mm\equiv{\bf I}$, i.e. when the communication channel is perfect. Specifically
%	%
%	%
%	% We show that the expected cost when communicating using a UDP-like protocol is higher than TCP-like. A system experiencing actuation packet losses will only achieve equality between the two protocols when~$\Mm = {\bf I}$ i.e. when the communication channel is perfect.
%	%
%	\beqn
%	J^*\left(\! \Gc_k \! \right)-J^*\left(\! \Fc_k \! \right) \ge 0.
%	\eeqn
		\end{theorem}
	\begin{proof}
	 The optimal control laws for each communication protocol are defined in (\ref{eq:opt-law-1}) and (\ref{eq:opt-law-2}).
%	\beqn
%	{\pazocal{U}^{^*}_k \! \left(\! \Fc_k \! \right)} &=& -\left(\Psi + 	\Omega_{g}\Vsb\right)^{\sminus1} \Omega_{gp} X_k, \nonumber \\
%	{\pazocal{U}^{^*}_k \! \left(\! \Gc_k \! \right)} &=& - \left(\Psi +\left({\bf I} \odot \Omega_{g}\right)\left({\bf I} -\Vsb\right) + \Omega_{g}\Vsb\right)^{\sminus1} \Omega_{gp} X_k.\nonumber 
%	\eeqn
Note that~${\Gm_{\Ic_k}\succ 0}$ and that,
	\beqn
	\Gm\!\!\left(\!\Gc_k\!\right)-\Gm\!\!\left(\!\Fc_k\!\right) \nquad \ &=& \nquad \  \left({\bf I} \odot \Omega_{g}\right)\left({\bf I} -\Vsb\right) \succ 0.\label{eq:G-inequality}
	\eeqn
Therefore,%
%	\beqn
%	\Gm\!\!\left(\!\Gc_k\!\right)&\succ& \Gm\!\!\left(\!\Fc_k\!\right). \nonumber 
%	\eeqn
~\cite[ 10.53]{seber} implies that
	\beqn
	\Gm^{\sminus1}\!\!\left(\!\Gc_k\!\right) \nquad \ &\prec& \nquad \  \Gm^{\sminus1}\!\!\left(\!\Fc_k\!\right),  \nonumber \\
%	\eeqn
%	By assumption~$\Mm\succ {\bf 0}$, therefore, it follows that
%	\beqn
	\Gm^{\sminus1}\!\!\left(\!\Gc_k\!\right)\Vsb \nquad \ &\prec& \nquad \  \Gm^{\sminus1}\!\!\left(\!\Fc_k\!\right)\Vsb, \nonumber \\
%	\eeqn
%	and therefore,
%	\beqn
%	X_k^{\sf T}\Omega_{gp}^{\sf T}\Vsb\Gm^{\sminus1}\!\!\left(\!\Gc_k\!\right)\Omega_{gp} X_k&<&X_k^{\sf T}\Omega_{gp}^{\sf T}\Vsb\Gm^{\sminus1}\!\!\left(\!\Fc_k\!\right)\Omega_{gp} X_k \nonumber \\
%	-X_k^{\sf T}\Omega_{gp}^{\sf T}\Vsb\Gm^{\sminus1}\!\!\left(\!\Gc_k\!\right)\Omega_{gp} X_k&>& -X_k^{\sf T}\Omega_{gp}^{\sf T}\Vsb\Gm^{\sminus1}\!\!\left(\!\Fc_k\!\right)\Omega_{gp} X_k \nonumber \\
		\Cm	-X_k^{\sf T}\Omega_{gp}^{\sf T} \Gm^{\sminus1}\!\!\left(\!\Gc_k\!\right)\Vsb\Omega_{gp} X_k \nquad \  &>& \nquad \  \Cm-X_k^{\sf T}\Omega_{gp}^{\sf T} \Gm^{\sminus1}\!\!\left(\!\Fc_k\!\right)\Vsb\Omega_{gp} X_k, \nonumber 
	\eeqn
where~$\Cm= X_k^{\sf T}\left( \Qm +\Omega_{p} \right)X_k + \tr\left(\Sigma_{\pazocal{W}}\Omega_{l}\right)$. Therefore, for any~${\bf 0}\prec \Mm \prec {\bf I}$ it holds that
\beqn
%	J^*\left(\! \Gc_k \! \right)&\ge& J^*\left(\! \Fc_k \! \right)\\
	J^*\left(\! \Gc_k \! \right)- J^*\left(\! \Fc_k \! \right) \nquad \ &>& \nquad \ 0.
\eeqn
	 This concludes the proof.
\end{proof}
%\hrule
%%Theorem~\ref{th:analytic-cost-increase} can also be utilised to show how the cost varies with respect to the channel statistics.
Additional insight can be obtained from Theorem~\ref{th:analytic-cost-increase}.
\begin{corollary}\label{cor:TCP>UDP}
It holds that %
\beqn
\left\| \pazocal{U}_k^* \! \left(\! \Gc_k \! \right)\right\|_{2}\nquad \ &<& \nquad \ \left\| \pazocal{U}_k^* \! \left(\! \Fc_k \! \right)\right\|_{2},
\eeqn
where~$\| \cdot \|_{2}$ denotes the 2-norm.
\end{corollary}
\begin{proof}
	 Starting with (\ref{eq:G-inequality}), it is seen that,
	\beqn
	\Gm\!\!\left(\!\Gc_k\!\right)-\Gm\!\!\left(\!\Fc_k\!\right) \nquad \ &=&\nquad \ \left({\bf I} \odot \Omega_{g}\right)\left({\bf I} -\Vsb\right) \succ 0, \\
	\Gm^{\sminus1}\!\!\left(\!\Gc_k\!\right) \nquad \ &\prec& \nquad \ \Gm^{\sminus1}\!\!\left(\!\Fc_k\!\right),  \\
	\Gm^{\sminus1}\!\!\left(\!\Gc_k\!\right)\Omega_{gp} X_k \nquad \ &<& \nquad \  \Gm^{\sminus1}\!\!\left(\!\Fc_k\!\right)\Omega_{gp} X_k,  \\
	\left\|\Gm^{\sminus1}\!\!\left(\!\Gc_k\!\right)\Omega_{gp} X_k\right\|_{2}\nquad \ &<&\nquad \  \left\|\Gm^{\sminus1}\!\!\left(\!\Fc_k\!\right)\Omega_{gp} X_k\right\|_{2},  \\
\left\| -\pazocal{U}_k^* \! \left(\! \Gc_k \! \right)\right\|_{2} \nquad \  &<& \nquad \  \left\| -\pazocal{U}_k^* \! \left(\! \Fc_k \! \right)\right\|_{2},  \\
\left\| \pazocal{U}_k^* \! \left(\! \Gc_k \! \right)\right\|_{2} \nquad \ &<& \nquad \ \left\| \pazocal{U}_k^* \! \left(\! \Fc_k \! \right)\right\|_{2}.
	\eeqn
%
%where~$\| \cdot \|_{2}$ denotes the 2-norm. %
%
%Note that for a vector the 2-norm is equivalent to the~$L_2$-Norm. %
%
This concludes the prof. %
\end{proof}

%As is seen in Corollary~\ref{cor:TCP>UDP} the TCP-like control law has the larger norm of the two control laws. %
%
The TCP-like protocol achieves a lower quadratic cost by using larger control signals to drive the states to zero quicker than the UDP-like protocol. %
%
% to a control law where the actuations at each time step have a greater magnitude, meaning that the TCP-like protocol in general has a shorter transient period when compared to the UDP-like control law. %
%
This difference is a result of the larger information set that the TCP-like protocol has access to.	
%
%This behaviour is shown within the numerical simulations later on. %
%

%The cost functions for each respective protocol are monotonically decreasing in cost in~$\Mm$. %
%
%For each protocol a matrix~$\Gm_{\Ic_k}$ has been defined that which decides the expected control cost. %
%
The following theorem shows that the cost function is monotonically decreasing functions in~$\Mm$.
\begin{theorem}\label{th:monotonic-cost-increase}
	Let~${\Mm_1 \in S^{m}_{++}}$ and~${\Mm_2 \in S^{M}_{++}}$ be diagonal matrices. If~${\Mm_{2} \succ\Mm_{1}}$ then %
	\beqn
		J^*_{\Delta\Mm} = J^*_{\Mm_1}\left(\! \Ic_k \! \right) - J^*_{\Mm_2}\left(\! \Ic_k \! \right) >0, 
	\eeqn
	where~${J^*_{\Mm_1}\left(\! \Ic_k \! \right)}$ is the optimal expected cost obtained with the value of~$\Mm_i$, where~$\Mm_i$ as the mean of the channel transmission variable,~$\Vm_k$. % with which the optimal expected cost has been calculated with respect to.
\end{theorem}

\begin{proof}
The proof is constructed for the TCP-like protocol. %
However, with the substitutions of~$\Omega_{g}$ for~$\Omega_{h}$ and~$\Psi$ for~$\left({\bf I}\odot \Omega_{g}\right) +\psi$ the corresponding UDP-like proof is identical. %
For a given~$\Mm_1$ and~$\Mm_2$ the cost difference between the optimal expected costs calculated for each~$\Mm_i$ respectively is %
\beqn
	J^*_{\Mm_1}\left(\! \Ic_k \! \right)\nquad  &=&\nquad \Cm	-X_k^{\sf T}\Omega_{gp}^{\sf T}   \left(
	\Vsb_1\Omega_{g} +\Psi\right)^{\sminus1}    \Vsb_1\Omega_{gp} X_k, \\
	J^*_{\Mm_2}\left(\! \Ic_k \! \right) \nquad &=& \nquad \Cm	-X_k^{\sf T}\Omega_{gp}^{\sf T} \left(
	\Vsb_2\Omega_{g} +\Psi\right)^{\sminus1}  \Vsb_2\Omega_{gp} X_k,  
\eeqn
where~$\Vsb_1$ and~$\Vsb_2$ are the diagonal matrices constructed from the matrices~$\Mm_1$ and~$\Mm_2$, such that~$\Vsb_{i} = {\bf I}_{N} \otimes \Mm_{i}$. %
Additionally, the constant~$\Cm $ is defined as~$\Cm= X_k^{\sf T}\left( \Qm +\Omega_{p} \right)X_k + \tr\left(\Sigma_{\pazocal{W}}\Omega_{l}\right)$. %
Consequently,~$\Vsb_2\succ\Vsb_1$ due to the assumption~$\Mm_2\succ\Mm_1$. %
Therefore, the cost difference between the two optimal expected costs is %
\beqn
	J^*_{\Delta\Mm}\nquad \  &=& \nquad \ J^*_{\Mm_1}\left(\! \Ic_k \! \right) - J^*_{\Mm_2}\left(\! \Ic_k \! \right) \\
	&=& \nquad  \Cm -X_k^{\sf T}\Omega_{gp}^{\sf T}   \left(
	\Vsb_1\Omega_{g} +\Psi\right)^{\sminus1}    \Vsb_1\Omega_{gp} X_k \nonumber \\
	&& - \left( \Cm - X_k^{\sf T}\Omega_{gp}^{\sf T} \left(
	\Vsb_2\Omega_{g} +\Psi\right)^{\sminus1}  \Vsb_2\Omega_{gp} X_k \right), \nonumber \\
	&=&\nquad  X_k^{\sf T}\! \Omega_{gp}^{\sf T} \!\! \left[\!\! \left( \!
	\Vsb_2\Omega_{g}\! +\! \Psi \! \right)^{\!\! \sminus1}  \!\!  \Vsb_2 \! - \! \left(\!
	\Vsb_1\Omega_{g} \! + \! \Psi \! \right)^{\!\! \sminus1} \!\! \Vsb_1 \!\! \right] \!\! \Omega_{gp} X_k , \nonumber \\
	&=&\nquad  X_k^{\sf T}\Omega_{gp}^{\sf T} \left(\Vsb_2\Omega_{g} +\Psi\right)^{\sminus1}\Vsb_2 \left[ \Vsb_1^{\sminus1}\left(
\Vsb_1\Omega_{g} +\Psi\right)\right. \nonumber \\
	&& \left. -  \Vsb_2^{\sminus1}\left(\Vsb_2\Omega_{g} +\Psi\right)   \right]\left(\Vsb_1\Omega_{g} +\Psi\right)^{\sminus1}  \Vsb_1 \Omega_{gp} X_k , \nonumber \\
	&=&\nquad  X_k^{\sf T}\Omega_{gp}^{\sf T} \left(\Vsb_2\Omega_{g} +\Psi\right)^{\sminus1}\Vsb_2 \nonumber \\
	&&  \times \left(\Vsb_1^{\sminus1} -  \Vsb_2^{\sminus1}\right)\Psi\left(\Vsb_1\Omega_{g} +\Psi\right)^{\sminus1}  \Vsb_1 \Omega_{gp} X_k . \label{eq:monotonic-cost-pos-def}
\eeqn
Only the term~$\Vsb_1^{\sminus1} -  \Vsb_2^{\sminus1}$ within~(\ref{eq:monotonic-cost-pos-def}) determines the positivity of the expected cost difference. %
The term is positive if~$\Mm_2\succ\Mm_1$, as is assumed above. %
Therefore,~${\Vsb_2\succ\Vsb_1}$ and~${\Vsb_1^{\sminus1} -  \Vsb_2^{\sminus1}\succ 0}$ and the expected cost difference is strictly positive. %
This concludes the proof. %
\end{proof}

\begin{corollary}\label{cor:com-cost}
	Theorem~\ref{th:analytic-cost-increase} combined with Theorem~\ref{th:monotonic-cost-increase} implies that with both protocols operating at a fixed cost value, the TCP-like protocol communicates with a larger packet loss rate. %
	\beqn
	J^*_{\Mm_1}\left(\! \Gc_k \! \right) = J^*_{\Mm_2}\left(\! \Fc_k \! \right),
	\eeqn
	where~$\Mm_{1} \succ \Mm_{2}$. %
\end{corollary}

\begin{proof}
	From Theorem~\ref{th:analytic-cost-increase} %
	\beqn
	J^*_{\Mm_1} \left(\Gc_k\right) = J^*_{\Mm_1}\left(\Fc_k\right) + \epsilon , \nonumber 
	\eeqn
	where $\epsilon \in \R^{+} $. %
	From Theorem~\ref{th:monotonic-cost-increase} it is known that
	\beqn
	J^*_{\Mm_1} \left(\Fc_k\right) + \epsilon = J^*_{\Mm_2}\left(\Fc_k\right)  .\nonumber 
	\eeqn
	Therefore,
	\beqn
	J^*_{\Mm_1}\left(\! \Gc_k \! \right) = J^*_{\Mm_2}\left(\! \Fc_k \! \right),\nonumber 
	\eeqn
	This concludes the proof. %
\end{proof}

\begin{remark}
	Theorem~\ref{th:monotonic-cost-increase} and Corollary~\ref{cor:com-cost} also apply to a non-stationary communication channelwith a slight adjustment of the conditions. %
	Both are true for a non-stationary channel when the stronger condition~$\Vsb_2\succ\Vsb_1$ holds, or more precisely,~$\Mm_{2\ k}\succ \Mm_{1\ k}$ for all~$k$. %
	%
%	Similarly, Corollary~\ref{cor:com-cost} is true for a non-stationary channel when the condition~${\Mm_{UDP\ k} \succ \Mm_{TCP\ k}}$ holds for all~$k$. %
\end{remark}

As is shown in Theorem~\ref{th:analytic-cost-increase}, the cost difference between the UDP-like and the TCP-like protocol is strictly positive for a channel without deterministic packet transmissions. %
The cost difference is zero only in the cases of no communication,~$\Mm={\bf 0}$, or perfect communication,~$\Mm = {\bf I}$. %

\subsection{Scalar Communication Channel}

The maximum difference in the expected cost and the maximising value of the expected packet transmission variable is characterised when the expected cost difference, as established in Theorem~\ref{th:analytic-cost-increase}, is simplified to the scalar case i.e.~$\Vsb\in\vR$. %
In doing so, the channel is simplified to a single channel that all actuators share. %
Additionally, the following results do not apply to a non-stationary communication channel. %

Assuming the same plant dynamics as~(\ref{eq:plant}), the cost difference between the two protocols as a function of~$\Vsb$ is given by
\beqn\label{eq:cost-diff-def}
J^*_{\Delta}\left(\Vsb\right) \nquad \ &\eqdef&  \nquad \ J^*\left(\! \Gc_k \! \right) - J^*\left(\! \Fc_k \! \right).
\eeqn
Note that Theorem~\ref{th:analytic-cost-increase} states that~(\ref{eq:cost-diff-def}) is positive, and therefore, the cost difference is
\beqn
J^*_{\Delta}\left(\Vsb\right)\nquad \ &=&\nquad \ \Cm	-X_k^{\sf T}\Omega_{gp}^{\sf T}\Vsb\Gm^{\sminus1}\!\!\left(\!\Gc_k\!\right)\Omega_{gp} X_k \nonumber \\
&& \qquad\qquad - \left( \! \Cm-X_k^{\sf T}\Omega_{gp}^{\sf T}\Vsb\Gm^{\sminus1}\!\!\left(\!\Fc_k\!\right)\Omega_{gp} X_k \! \right) \ \nonumber \\
%
% &=& X_k^{\sf T}\Omega_{gp}^{\sf T}\Vsb\left(\Gm^{\sminus1}\!\!\left(\!\Fc_k\!\right)-\Gm^{\sminus1}\!\!\left(\!\Gc_k\!\right)\right)\Omega_{gp} X_k\nonumber \\
%
% &=& X_k^{\sf T}\Omega_{gp}^{\sf T}\Vsb\Gm^{\sminus1}\!\!\left(\!\Gc_k\!\right)\left(\Gm\!\!\left(\!\Gc_k\!\right)-\Gm\!\!\left(\!\Fc_k\!\right)\right)\Gm^{\sminus1}\!\!\left(\!\Fc_k\!\right)\Omega_{gp} X_k \nonumber \\
%
&=&  \ \nquad X_k^{\sf T}\Omega_{gp}^{\sf T} \! \Vsb \! \Gm^{\sminus1}\!\!\left(\!\Gc_k\!\right) \!\! \left(\! 1-\Vsb \! \right) \!\! \left({\bf I} \odot \Omega_{g} \! \right) \!\! \Gm^{\sminus1}\!\!\left(\!\Fc_k\!\right) \! \Omega_{gp} X_k  \ \nonumber \\
%
% &=& \tr \left(\Vsb\left(1-\Vsb\right)\Gm^{\sminus1}\!\!\left(\!\Gc_k\!\right)\left({\bf I} \odot \Omega_{g}\right)\Gm^{\sminus1}\!\!\left(\!\Fc_k\!\right)\Omega_{gp} X_k X_k^{\sf T}\Omega_{gp}^{\sf T} \right) \nonumber \\
%
&=& \ \nquad \Vsb\left(1-\Vsb\right) \tr \left(\Gm^{\sminus1}\!\!\left(\!\Gc_k\!\right)\left({\bf I} \odot \Omega_{g}\right)\Gm^{\sminus1}\!\!\left(\!\Fc_k\!\right)\Lm \right), \label{eq:cost-diff-eq}
\eeqn
where~$\Lm= \Omega_{gp} X_k X_k^{\sf T}\Omega_{gp}^{\sf T}$. %
From~(\ref{eq:cost-diff-eq}) it is seen that the cost difference between the protocols depends on a scaling of the variance of the packet transmission variable,~$\Vsb\left(1- \Vsb\right)$, over the prediction horizon,~$N$. %
Intuitively, this means that for a channel with a high variance the cost difference is larger. %
The TCP-like protocol has access to more information and is better able to reduce the uncertainty in the state caused by~$\Vs$ than the UDP-like protocol, and therefore, has a smaller cost. %
However, the difference is a non-linear function of~$\Vsb$ owing to the dependence of~$\Gm^{\sminus1}\!\!\left(\!\Fc_k\!\right)$ and~$\Gm^{\sminus1}\!\!\left(\!\Gc_k\!\right)$ on~$\Vsb$. %
At this point we characterise the maximum cost difference as a function of~$\Vsb$. %
This maximum cost difference corresponds to the greatest cost difference incurred by the operator choosing to communicate using a UDP-like communication protocol instead of a TCP-like protocol. %
Lemma~\ref{lem:cost-diff-deriv} presents this result.
\begin{lemma}\label{lem:cost-diff-deriv}
%	The cost difference between the UDP-like and the TCP-like protocols is defined as
%	%
%	\beqn
%	J_\Delta^*\left(\Vsb\right) \eqdef J^*\left(\! \Gc_k \! \right)- J^*\left(\! \Fc_k \! \right)&>&0.
%	\eeqn
%	%
	The derivative of the cost difference in~(\ref{eq:cost-diff-eq}) is
	\beqn
	{{\partial}\over{\partial \! \Vsb }}	J_\Delta^* \! \left( \!\Vsb \! \right)\nquad 
	&=&\nquad \  X_k^{\sf T}\Omega_{gp}^{\sf T} \left(
	\Gm^{\sminus1}\!\!\left(\!\Gc_k\!\right)\left(\vphantom{\Big|} \left(1-2\Vsb \right) \Omega_{d} \right.\right.\nonumber\\
	&&\nqquad\nqquad\nqquad \ \left.\left.  \sminus  \! \Vsb\!\! \left( \!\! 1 \sminus \! \Vsb \! \right)\!\!\! \left[\! \Omega_{h} \! \Gm^{\sminus1}\!\!\left(\!\Gc_k\!\right) \! \Omega_{d} \! + \! \Omega_{d} \Gm^{\sminus1}\!\!\left(\!\Fc_k\!\right) \! \Omega_{g} \!\! \right] \!\!  \right) \!\! \Gm^{\sminus1}\!\!\left(\!\Fc_k\!\right) \!\! \right) \!\! \Omega_{gp} X_k \! , \ \ \quad \label{eq:cost-diff-deriv}
	\eeqn
	where~$\Omega_{d}=\left({\bf I} \odot \Omega_{g}\right)$ and~$\Vsb \in \vR$.
\end{lemma}
\proof See Appendix.

Finding the critical points of the cost difference (\ref{eq:cost-diff-deriv}) is non-trivial for this function due to the outer products. %
%any two vectors has a rank equal to~$1$, it is therefore impossible to remove the terms from the quadratic. 
%%%%%%%%%%%%%%%%%%%%%%%%%%%%%%%%%%%%%
%Additionally, we have constructed the system thus far under the assumption that the initial conditions,~$X_k$, are unrestricted other than being non-zero. %
%%%%%%%%%%%%%%%%%%%%%%%%%%%%%%%%%%%%%%
In order to find the stationary points of (\ref{eq:cost-diff-deriv}) it is required that the maximum eigenvalue of the matrix inside the quadratic tends to~$0$. %
The maximum eigenvalue of a matrix can be written as~\cite{HornJohnMatrix}
\beqn
\lambda_{\max}\left(\Cm\right) = \max_{\|x\|\neq 0 } {{\xv^{\sf T}\Cm \xv}\over{\xv^{\sf T}\xv}},\label{eq:Rayleigh-Ritz}
\eeqn
where~$\xv \in \R^n$ is a column vector,~$\Cm$ is a square matrix of appropriate dimension, and~$\lambda_{\max}\left(\Cm\right)$ is the maximum eigenvalue of~$\Cm$. %
%
%In our scenario we want to ensure that the maximum eigenvalue is~$0$. %
%
Any matrix for which all eigenvalues are equal to~$0$ also satisfies that the determinant is zero. %
However, not all matrices with a~$0$ determinant have a maximum eigenvalue equal to~$0$. %
%
%This is the point at which we begin solving for the critical points.%
%
Solving for a zero determinant of~(\ref{eq:cost-diff-deriv}) results in a finite number of values for~$\Vsb$, at which point the condition~(\ref{eq:Rayleigh-Ritz}) reveals the critical points. %
This leads to the next theorem. 
\begin{lemma}\label{lem:Nu-Sols}  
		The equation
		\beqn
	&&	\det\left(\vphantom{\Big|}\Gm_{\Gm}\!\!\left(\Vsb\right)\left[ \left(1-2\Vsb \right) \Omega_{d} \right.\right. \nonumber \\
	&& \nqquad\nqquad \left.\left.
		-  \! \Vsb\!\! \left( \!\! 1 - \! \Vsb \! \right)\!\!\! \left[\! \Omega_{h} \! \Gm_{\Gm}\!\!\left(\Vsb\right) \Omega_{d} + \Omega_{d} \Gm_{\Fm}\!\!\left(\Vsb\right) \Omega_{g} \right] \right]\Gm_{\Fm}\!\!\left(\Vsb\right)\vphantom{\Big |}\!\!\right)\! = \! 0 \label{eq:det-0}
		\eeqn
		has~$2Nm$ solutions, given by,
		\begin{subequations}\label{eq:det-sol}
		\beqn
		\Vsb^D_{2i\sminus1} \nquad \  &=& \nquad \  {{1}\over{1 + \sqrt{1+\lambda_i}}},\\ 
		\Vsb^D_{2i} \nquad \  &=& \nquad \  {{1}\over{1 - \sqrt{1+\lambda_i}}},
		\eeqn	
		\end{subequations}
		where~$\Vsb^D_i~$ correspondences to the~$i$-th solution of~(\ref{eq:det-0}) and~$\lambda_i$ is the~$i$-th eigenvalue of the matrix,
		\beqn
		\left(\Omega_{g}\Omega_{d}^{\sminus1}\left(\Omega_{g} + \Psi\right) + \Psi \Omega_{d}^{\sminus1}\Omega_{h} \right) \left(\Omega_{g} + \Psi\right)^{\sminus1} \Omega_{d}\Psi^{\sminus1}.
		\eeqn
\end{lemma}

\proof See Appendix.

The above theorem gives a solution in~$\Vsb$ for all of the points for which~(\ref{eq:det-0}) holds true. %
However, as mentioned above this does not correspond to all of the critical points of~(\ref{eq:cost-diff-def}). %
In order for~$\Vsb^{D}_i$ to be a critical point of~(\ref{eq:cost-diff-def}) it must hold that the magnitude of the maximum eigenvalue of~(\ref{eq:cost-diff-deriv}) must also be~$0$. %
We address this by solving the following numerical evaluation problem.
\begin{theorem}\label{th:maximal-nu}
	The cost difference between the UDP-like and the TCP-like protocols as a function of~$\Vsb$, is defined as 
	\beqn
	J_\Delta^*\left(\Vsb\right) \eqdef J^*\left(\! \Gc_k \! \right)- J^*\left(\! \Fc_k \! \right)&>&0.
	\eeqn
	Has a maximum point that occurs at~$\Vsb^{D^{*}}$, where~$\Vsb^{D^{*}}$ is defined as
	\beqn	\label{eq:maximising-nu}
	\Vsb^{D^{*}}\nquad &\eqdef&\nquad  \sup_{\Vsb^{D}_i\in\vR} \ J_\Delta^*\left(\Vsb^{D}_i\right)  \ \mbox{s.t.}  
	 \left\|\max_{\left\|x\right\|\neq 0 } {{x^{\sf T} f\left(\Vsb^{D}_i\right)   x}\over{x^{\sf T}x}}\right\|\! =\! 0, \ \ 
	\eeqn
	and where~$f\left(\Vsb^{D}_i\right)$ is defined as
	\beqn
	f\left(\Vsb^{D}_i\right)\nquad \ &=& \nquad \ \Gm_{\Gm}\!\!\left(\Vsb^{D}_i\right)\left[ \left(1-2\Vsb^{D}_i \right) \Omega_{d} \right.\nonumber\\
	&&\nqquad \nqquad \nquad   \left.  \sminus \Vsb^{D}_i\!\! \left(1\! \sminus \! \Vsb^{D}_i\right)\!\! \! \left[\Omega_{h}\Gm_{\Gm}\!\!\left(\Vsb^{D}_i\right)\Omega_{d}\!\! +\! \Omega_{d} \Gm_{\Fm}\!\!\left(\Vsb^{D}_i\right) \Omega_{g}\right] \! \right]\!\! \Gm_{\Fm}\!\!\left(\Vsb^{D}_i\right)\!\!. \ \ \ \ \
	\eeqn
	
\end{theorem} 

\begin{proof}
Lemma~\ref{lem:Nu-Sols} states that every~$\Vsb^{D}_i$ results in the determinant in~(\ref{eq:det-0}) being equal to~$0$. %
However, this theorem does not guarantee that~$\Vsb^{D}_i$ is a critical point of~(\ref{eq:cost-diff-def}). %
It is also required that the magnitude of the maximum eigenvalue is~$0$ for a given~$\Vsb^{D}_i$. %
Therefore, the condition on the~$\Vsb^{D}_i$ is recast as
\beqn
&& \max_j\ \left|\lambda_j\!\!\left(\vphantom{\Big|}\Gm_{\Gm}\!\!\left(\Vsb^{D}_i\right)\left[ \left(1-2\Vsb^{D}_i \right) \Omega_{d} \right.\right.\right.\nonumber\\
&&\nqquad  \!\!\! \left.\left.\left.\vphantom{\Big|} \sminus \Vsb^{D}_i\!\! \left(1\! \sminus \! \Vsb^{D}_i\right)\!\! \! \left[\Omega_{h}\Gm_{\Gm}\!\!\left(\Vsb^{D}_i\right)\Omega_{d}\!\! +\! \Omega_{d} \Gm_{\Fm}\!\!\left(\Vsb^{D}_i\right) \Omega_{g}\right] \! \right]\!\! \Gm_{\Fm}\!\!\left(\Vsb^{D}_i\right)\!\right)  \right|\!\! = 0.\nonumber
\eeqn
%
%It follows from~(\ref{eq:Rayleigh-Ritz}) that the maximum eigenvalue of a matrix~$A$ can be expressed as
%%
%\beqn
%\lambda_{\max} = \max_{\|x\| \neq 0 } {{\xv^{\sf T}A \xv}\over{\xv^{\sf T}\xv}}.\nonumber 
%\eeqn
%%
Therefore, the condition to ensure that~$\Vsb^{D}_i$ is a critical point becomes
\beqn
 \left|\max_{\left\|x\right\|\neq 0 } {{x^{\sf T}f\left(\Vsb^{D}_i\right) x}\over{x^{\sf T}x}}\right| = 0,
\eeqn
where
\beqn
	f\left(\Vsb^{D}_i\right)\nquad \ &\eqdef& \nquad \ \Gm_{\Gm}\!\!\left(\Vsb^{D}_i\right)\left[ \left(1-2\Vsb^{D}_i \right) \Omega_{d} \right.\nonumber\\
&&\nqquad \nqquad \nquad   \left.  \sminus \Vsb^{D}_i\!\! \left(1\! \sminus \! \Vsb^{D}_i\right)\!\! \! \left[\Omega_{h}\Gm_{\Gm}\!\!\left(\Vsb^{D}_i\right)\Omega_{d}\!\! +\! \Omega_{d} \Gm_{\Fm}\!\!\left(\Vsb^{D}_i\right) \Omega_{g}\right] \! \right]\!\! \Gm_{\Fm}\!\!\left(\Vsb^{D}_i\right)\!\!. \ \ \ \ \
\eeqn
Theorem~\ref{th:analytic-cost-increase} states that in~$\vR$ the cost difference is strictly positive. Therefore, there is at least one maximum in this interval. Taking the supremum of all critical points that lie within~$\vR$ results in the maximising~$\Vsb^{D}_i$ in~$\vR$. 
This is denoted by~$\Vsb^{D^{*}}$. %
This concludes the proof.
\end{proof}
\section{Case Studies}\label{sec:Numerical-Results}
We conduct case studies for two separate systems. %
\subsection{Single Actuator System}\label{sec:pendulum}
The first system we consider is the inverted pendulum system as used in~\cite{1}. The discrete time state space model of the pendulum is provided in~\cite{1} and reported here for convenience
\begin{subequations}\label{eq:pendulum-ss}
\beqn
\Am\nquad \ &=&\nquad \ \left(\!\!\!\!\ba{cccc}
 1.001 & 0.005  & 0.000  & 0.000 \\
 0.35  & 1.001  & -0.135 & 0.000 \\
-0.001 & 0.000  & 1.001  & 0.005 \\ 
-0.375 & -0.001 & 0.590  & 1.001
 \ea\!\!\right), \\
  \Bm \nquad \ &=& \nquad \ \left(\!\!\ba{c} 
  0.001 \\
  0.540 \\
 -0.002 \\
 -1.066 \ea\!\!\right), \quad
\Psi = 2{\bf I}_{N}, \\
\Omega \nquad \ &=& \nquad \   {\bf I}_{N} \otimes \left(\!\!\!\!\ba{cccc}
5 & 0 & 0 & 0 \\
0 & 1 & 0 & 0 \\
0 & 0 & 1 & 0 \\ 
0 & 0 & 0 & 1
\ea\!\!\right) ,
\eeqn
\end{subequations}
the prediction horizon is~$N=80$, and~$\otimes$ denotes the Kronecker product. %
Note that this system has a single actuator, and therefore, there is no difference between the variables~$\mu$,~$\Mm$, and~$\Vsb$. %
Indeed, they are all equivalent to scalars multiplied by an appropriately dimensioned identity matrix. %xcept the size of the identity matrix which they are multiplied by. %
In the below the discussion revolves around the variable~$\Vsb$, however, this is interchangeable with either of the other variables for this system. %

The analysis in Section \ref{sec:Numerical-Results} characterises the maximal cost difference between the two protocols. Therein, the system is reduced to a scalar communication channel as in the communication channel in~\cite{1}. %
The expected cost difference as a function of the packet transmission parameter,~$\Mm$, is described by Theorem~\ref{th:analytic-cost-increase} to converge to~$0$ for~$\Mm\rightarrow {\bf I}$ or~$\Mm\rightarrow {\bf 0}$. %
The convergence of the cost difference in these limit cases is observed in Fig.~\ref{fig:log-cost-pendulum}. %   and Fig.~\ref{fig:cost-diff-pendulum}
The limit cases of~$\Mm\rightarrow {\bf I}$ or~$\Mm\rightarrow {\bf 0}$ correspond to deterministic cases. % of~$\Vs$. %
Note, these are the only cases for which the TCP-like and UDP-like information sets are equivalent, and therefore, the control laws are identical. %
Furthermore, in the limit case when~$\Mm\rightarrow{\bf I}$ the control law for both protocols are also equivalent to a nominal LQG controller that does not account for packet loss in the actuation channel. %
It is seen in Figure~{\ref{fig:log-cost-pendulum} that the characterisation given in Section \ref{sec:Cost-Difference-Analysis} corresponds to the observed behaviour. %
Specifically, there is a maximal point within the~$\vR$ region. %
This maximal point is highlighted with the vertical arrow, which corresponds to the maximal point,~$\Vsb^{D^*} = 0.0031$, as predicted by Theorem~\ref{th:maximal-nu}. %	
%
%
%
%When considering the optimal channel efficiency it should be noted that due to~(\ref{eq:pendulum-ss}) having a single actuator the solution of~(\ref{eq:channel-opt}) is a single scalar point. Specifically the solution is the value of~$\mu$ that achieves a cost of~$\alpha$ with equality for each of the protocols. %
%
%Additionally, due to having a single actuator the weighting term is irrelevant for this system as there is no trade off between communication channels.
%
Theorem~\ref{th:analytic-cost-increase} states that the TCP-like cost is strictly less than the UDP-like as is also seen in Fig.~\ref{fig:log-cost-pendulum}. %
Additionally, as seen in Corollary~\ref{cor:com-cost}, this means that for a given expected system cost,~$J^*\left(\! \Ic_k \! \right)$, the TCP-like protocol has a smaller channel transition probability,~$\Vsb$. %
This is seen within Fig.~\ref{fig:log-cost-pendulum}.
%
%Specifically, Theorem~\ref{th:analytic-cost-increase} implies that for a fixed~$\alpha$~$ J^*_C\left(\! \Gc_k \! \right) > J^*_C\left(\! \Fc_k \! \right)$. %
%
%This suggests that for system where it is costly to communicate an operator would be better suited to utilise TCP-like protocols rather than UDP-like protocols within the control system. % 
%
%

\begin{figure}[!t]
	\captionsetup{justification=centering,margin=0cm,width=\linewidth}
	\centering
\includegraphics[width=1\columnwidth]{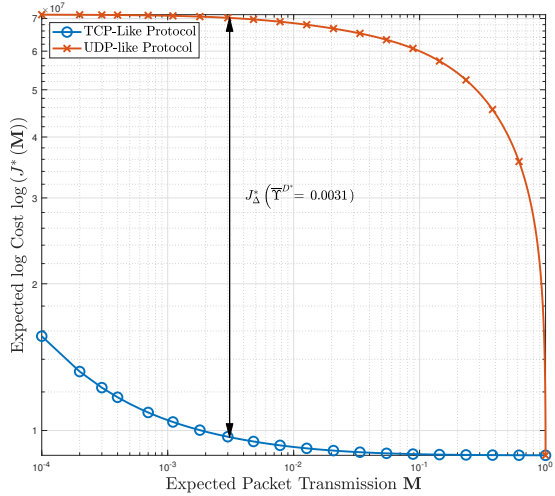}
	\caption[center]{The log cost plot of the TCP-like and the UDP-like protocols for the inverted pendulum system seen in~\cite{1}. The mean of the packet transmission variable,~${\bf {V}}_k$, is varied over the~$\vR$ region.}\label{fig:log-cost-pendulum}
\end{figure}

%
%\begin{figure}[!t]
%	\captionsetup{justification=centering,margin=2cm,width=\linewidth}
%	\centering
%	\includegraphics[width=1.1\columnwidth]{Figures/log_cost_diff_pendulum.eps}
%	\caption[center]{\textcolor{red}{This plot needs some lines to show possible maximums as stated in the theorems in addition to a box around the axis.} The log cost plot of the difference in cost of the UDP-like and the TCP-like protocol for the inverted pendulum system seen in~\cite{1}. Where the  mean of the packet loss variable,~$\Mm$, is varied over the~$\vR$ region.}\label{fig:cost-diff-pendulum}
%\end{figure}
%

The cost difference of the two protocols arises from different control laws. %
Table~\ref{tab:closed-loop-eig} shows the closed loop eigenvalues of the TCP-like and the UDP-like protocols. %
Where the closed loop gain~$\Km_{\Ic_k}$ is defined as the first~$m$ by~$n$ block of the matrix~$\Gm_{\Ic_k}\Omega_{gp}~$. %
Additionally, the packet transmission variable is set at~$\Vsb = 0.9$ for the calculation of the closed loop eigenvalues. %
As shown in Table~\ref{tab:closed-loop-eig}, the TCP-like protocol has a conjugate pair of eigenvalues with a smaller complex component when compared to the UDP-like eigenvalues. %
This suggests the damping of the state response is lower for the UDP-like protocol than the TCP-like protocol. %
%
%Additionally, the real component of the complex conjugate is closer to~$0$ for the TCP-like protocol. %
%
Additionally, the magnitude of the TCP complex-conjugate eigenvalues is closer to the origin, this points to faster decay in the response of these modes.
%
%the state response of the UDP-like protocol is more damped when compared to the TCP-like protocol.
%
%the TCP-like control law has closed loop eigenvalues closer to~$0$. %
%
%Both protocols have a conjugate pair of complex closed loop eigenvalues. %
%
%Note that the closed-loop eigenvalues correspond to a stable system for both protocols, which is expected from the results in~\cite{1}.
%
\newcolumntype{d}[1]{>{\raggedright\let\newline\\\arraybackslash\hspace{0pt}}p{#1}}
\renewcommand*{\arraystretch}{1.2}
%\hspace{-0.1cm}
\begin{table}[!t]
	\caption{Closed Loop Eigenvalues of (\ref{eq:pendulum-ss})}\label{tab:closed-loop-eig}
	\centering
	\begin{tabular}{|c|d{2.39cm}|d{2.39cm}|}
		\hline
	$\lambda_i(\Am - \Bm \Km_{\Ic_k})$ & TCP-like Eigenvalues & UDP-like Eigenvalues\\
	\hhline{|=|=|=|}
	$\lambda_1$ &~$0.9497 + 0.0056i$		&~$0.9907 + 0.0201i$ \\\hline
	$\lambda_2$ &~$0.9497 - 0.0056i$		&~$0.9907 - 0.0201i$ \\\hline
	$\lambda_3$ &~$-0.1148$					&~$0.9729$ \\\hline
	$\lambda_4$ &~$0.9978$					&~$0.9382$ \\\hline
	\end{tabular}
\end{table}
\subsection{Multiple Actuator System}\label{sec:multi-act-cs}
For the second case study we consider an arbitrary system with multiple actuators. This system is constructed as,
%[1.03 0.005;0.35 1.01]
\begin{subequations}\label{eq:mixed-ss}
\beqn 
\Am \nquad \ &=&\nquad \ \left(\!\!\ba{cc}
 1.03 & 0.005 \\
 0.35 & 1.01 \ea\!\!\right), \quad 
\Bm = \left(\!\!\ba{cc}
 1 & 0 \\
 0 & 1 \ea\!\!\right), \\
\Psi \nquad \ &=& \nquad \  {\bf I}_{N} \otimes \left(\!\!\ba{cc}
 1 & 0 \\
 0 & 1 \ea\!\!\right) , \quad\ \ \ 
\Omega =  {\bf I}_{N} \otimes \left(\!\!\ba{cc}
 1 & 0 \\
 0 & 1 \ea\!\!\right) ,
\eeqn
\end{subequations}
and the prediction horizon is set to~$N=10$. The system in~(\ref{eq:mixed-ss}) has multiple actuators and thus is used to highlight the generality of our results. %
\begin{figure}[!t]
	\captionsetup{justification=centering,margin=2cm,width=\linewidth}
	\centering
	\includegraphics[width=\columnwidth]{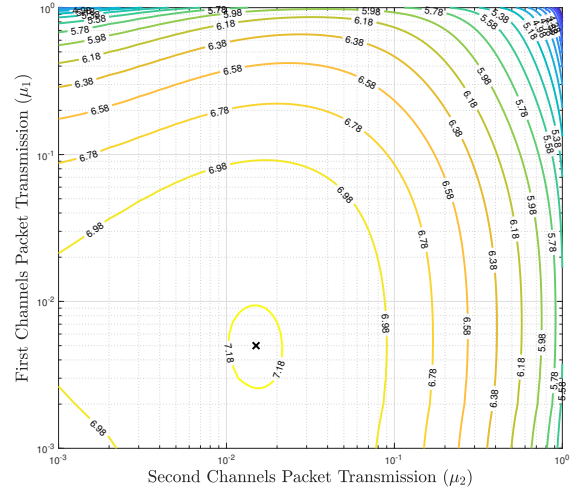}
	\caption[center]{The log cost contour plot of the difference in cost of the UDP-like and the TCP-like protocol for the system~(\ref{eq:mixed-ss}). The mean of the packet transmission variable,~$\Mm$, is varied over the~$\vR^{2}$ region. The point of maximal cost difference is marked with a black cross.}\label{fig:cost-diff-mixed}
\end{figure}
%
%Although the theorems developed in Section~\ref{sec:Cost-Difference-Analysis} don't apply directly to the system~(\ref{eq:mixed-ss}). %
%
The expected cost difference, as seen in Fig.~\ref{fig:cost-diff-mixed}, shows existence of a maximal point despite the system having multiple actuators, as predicted by Theorem~\ref{th:analytic-cost-increase}. %
This is marked with a black cross in Fig.~\ref{fig:cost-diff-mixed}. %
This indicates that the results extend to multiple actuators.
Fig.~\ref{fig:cost-diff-mixed} shows that the cost difference is strictly positive when~$\Mm \neq {\bf I}$ and~$\Mm \neq {\bf 0}$. %

When considering the expected cost for this system, the presence of a second actuator means the expected cost is a function of multiple packet transmission variables. %
As a result, there are regions within~$\vR^2$ where the expected cost remains fixed for a range of values of~$\Mm$. %
This behaviour is depicted in Fig.~\ref{fig:TCP-log-cost-mixed} and Fig.~\ref{fig:UDP-log-cost-mixed} for the TCP-like and UDP-like expected costs, respectively.
%
%
%It is interesting to note the shape of the optimal efficiency plots. Specifically, the differences between the TCP-like and the UDP-like optimal efficiency plots. The TCP-like optimal efficiency plot is linear for the majority of the the region other than a non-linear region towards the origin. Contrasting this the UDP-like optimal efficiency plot is non-linear over the whole region. %
%
%In addition to these differences there are similarities in the optimal efficiencies. %
%
\begin{table}[!t]
	\caption{Closed Loop Eigenvalues of (\ref{eq:mixed-ss})}\label{tab:closed-loop-eig-mixed}
	\centering
	\begin{tabular}{|c|d{2.39cm}|d{2.39cm}|}
		\hline
		$\lambda_i(\Am - \Bm \Km_{\Ic_k})$ & TCP-like Eigenvalues & UDP-like Eigenvalues\\
		\hhline{|=|=|=|}
		$\lambda_1$ &~$-0.1082$				&~$0.4904 + 0.0312i$ \\\hline
		$\lambda_2$ &~$-0.8938$				&~$0.4904 - 0.0312i$ \\\hline
	\end{tabular}
\end{table}
As mentioned for the previous case study, the control laws developed in Section \ref{sec:MPC-Optimal-Control} are different for each protocol. %
As shown in Table~\ref{tab:closed-loop-eig-mixed}, the UDP-like protocol has a complex conjugate pair of eigenvalues. % suggesting oscillatory behaviour. %
%
%Interestingly, the TCP-like protocol does not have a smaller real part than the UDP-like for both closed-loop eigenvalues. %
%
Note that for the purposes of calculating the eigenvalues in Table~\ref{tab:closed-loop-eig-mixed}, the communication channel is not a scalar, in fact we set~$\Mm=\usebox{\smlV}$.
%
%Namely, the TCP-like protocol induces a rapid response towards~$0$ for all states, whereas the UDP-like protocol displays a damped response. %
%
As with the pendulum case study it is seen that the TCP-like and the UDP-like control laws induce different behaviour in the state trajectories. %
%
%\begin{figure}[!t]
%	\captionsetup{justification=centering,margin=2cm,width=\linewidth}
%	\centering
%	\includegraphics[width=1.1\columnwidth]{Figures/average_state_mixed.eps}
%	\caption[center]{The state trajectories for the semi unstable system. Where the mean of the packet loss variable,~$\Mm$, is set at~$\usebox{\smlV}$; the prediction horizon,~$N$ is set at~$10$; and the simulation is averaged over~$1000$ realisations.}\label{fig:averaged-states-mixed}
%\end{figure}
%%
%\begin{figure}[!t]
%	\captionsetup{justification=centering,margin=2cm,width=\linewidth}
%	\centering
%	\includegraphics[width=1.1\columnwidth]{Figures/Normalised_Cumulative_Cost_mixed.eps}
%	\caption[center]{The normalised cumulative cost for the unstable system seen in~(\ref{eq:mixed-ss}). Where the mean of the packet loss variable,~$\Mm$, is set at~$\usebox{\smlV}$; the prediction horizon,~$N$ is set at~$10$; and the simulation is averaged over~$1000$ realisations.}\label{fig:normalised-cumulative-cost-mixed}
%\end{figure}
%
\begin{figure}[!t]
	\captionsetup{justification=centering,margin=2cm,width=\linewidth}
	\centering
	\includegraphics[width=\columnwidth]{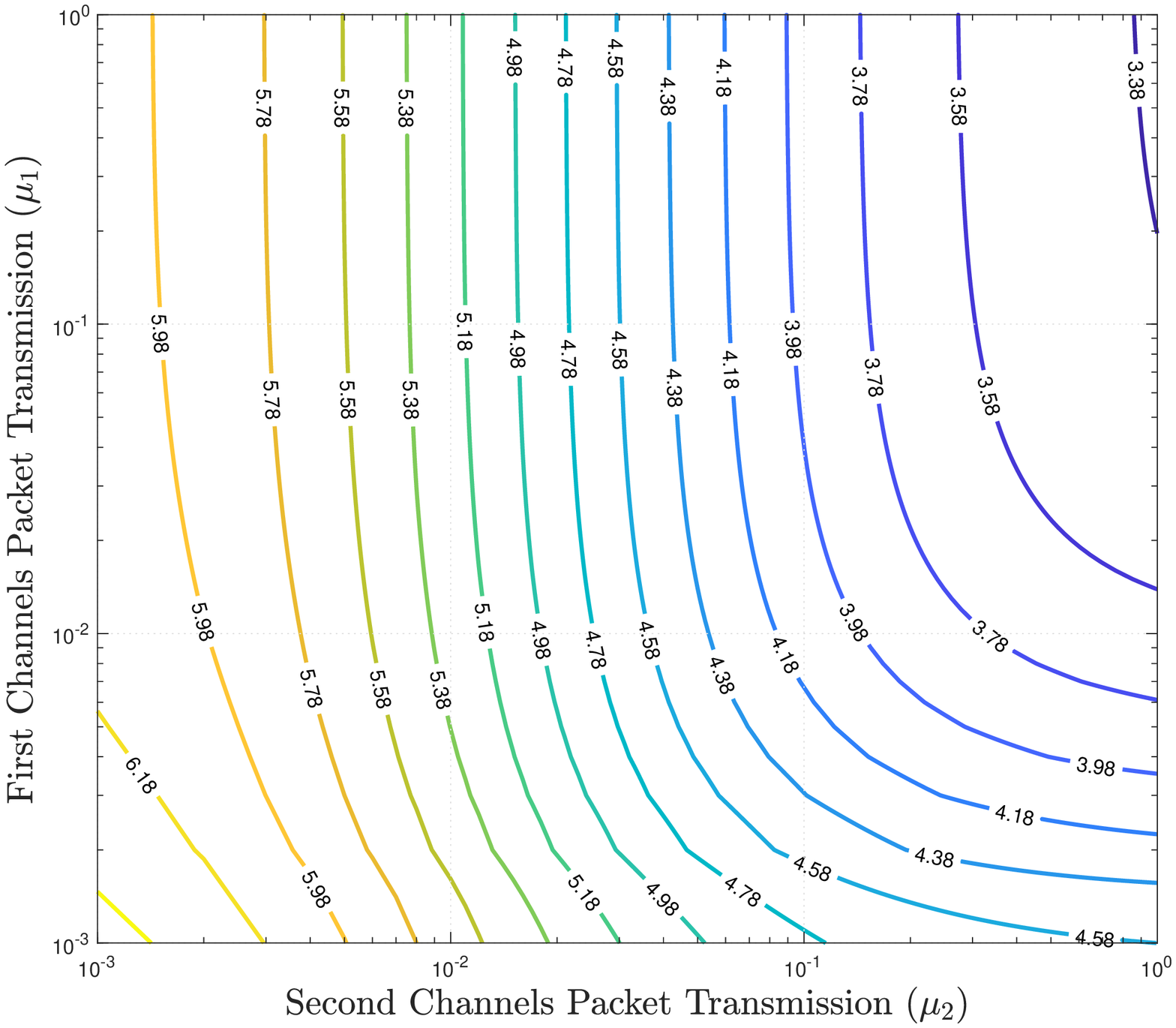}
	\caption[center]{The log expected cost contour plot of the TCP-like protocol for the system~(\ref{eq:mixed-ss}). The mean of the packet transmission variable,~$\Mm$, is varied over the~$\vR^{2}$ region.}\label{fig:TCP-log-cost-mixed}
\end{figure}
\begin{figure}[!t]
	\captionsetup{justification=centering,margin=2cm,width=\linewidth}
	\centering
	\includegraphics[width=\columnwidth]{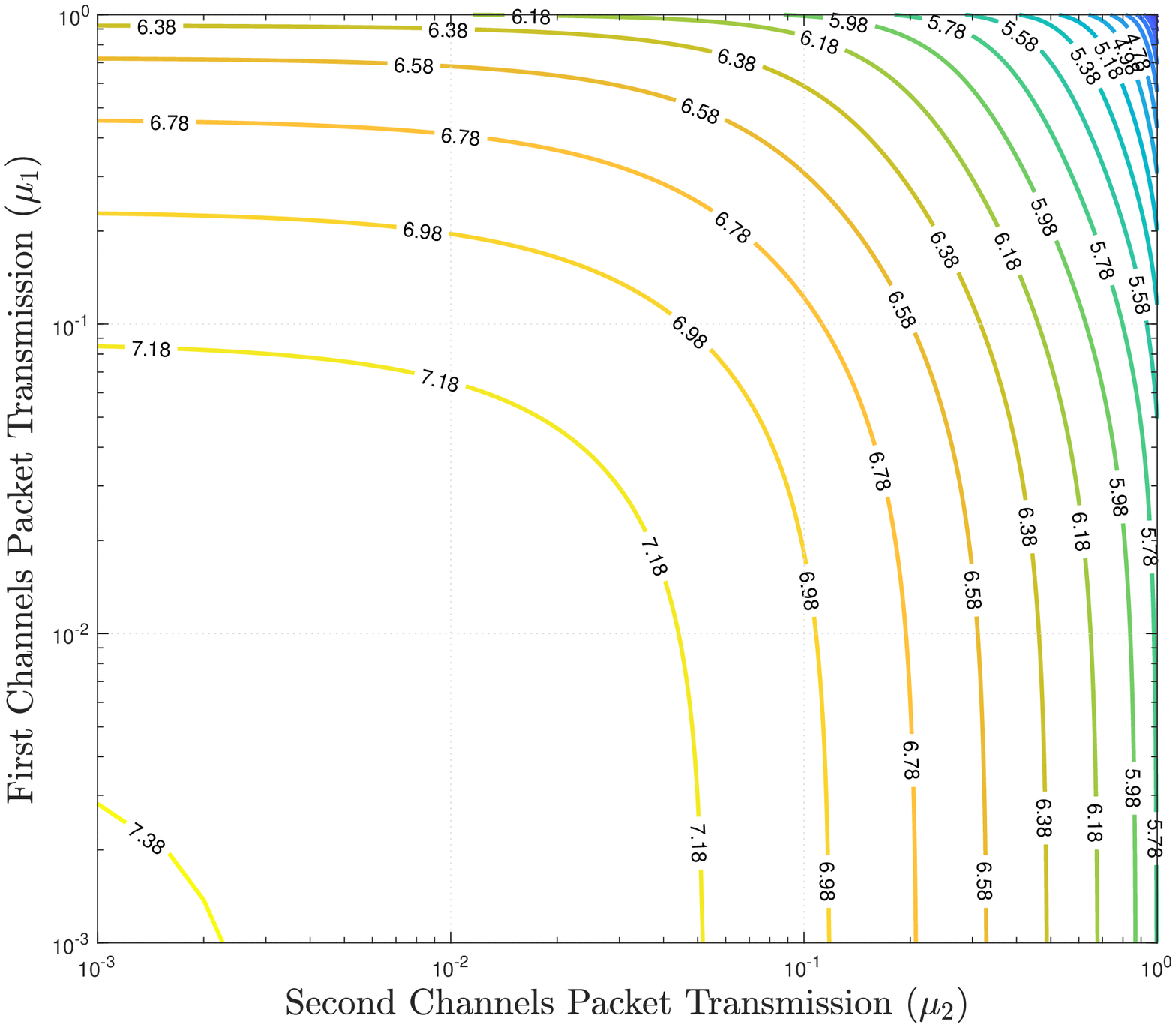}
	\caption[center]{The log expected cost contour plot of the UDP-like protocol for the system~(\ref{eq:mixed-ss}). The mean of the packet transmission variable,~$\Mm$, is varied over the~$\vR^{2}$ region.}\label{fig:UDP-log-cost-mixed}
\end{figure}

\section{Packet Loss Allocation Optimisation}\label{sec:channel-disc} %IE April 2021

\begin{figure}[!t]
	\captionsetup{justification=centering,margin=2cm,width=\linewidth}
	\centering
	\includegraphics[width=\columnwidth]{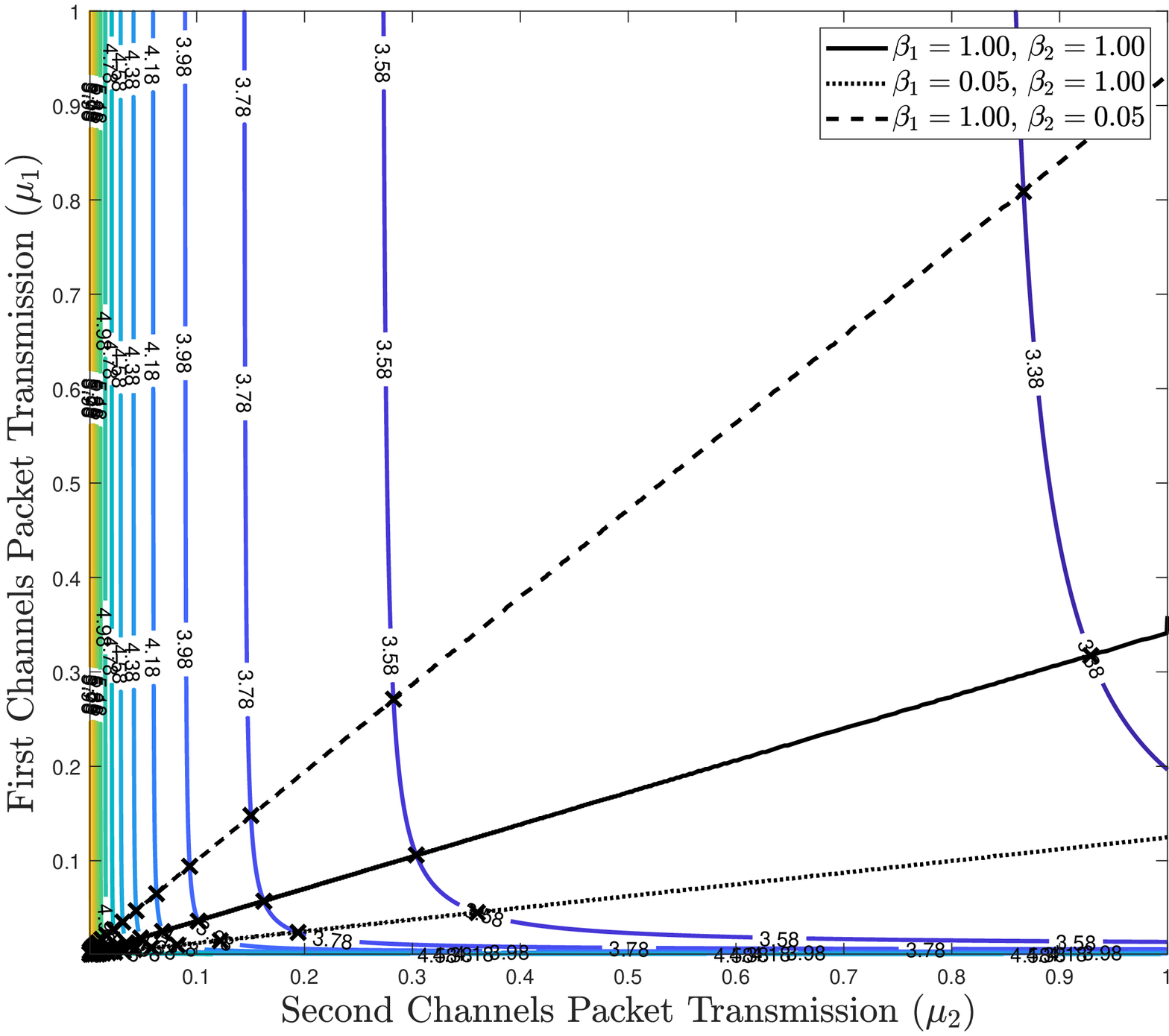}
	\caption[center]{The log expected cost contour plot of the TCP-like protocol for the system~(\ref{eq:mixed-ss}). The mean of the packet transmission variable,~$\Mm$, is varied over the~$\vR^{2}$ region.}\label{fig:TCP-log-cost-mixed-mult}
\end{figure}
\begin{figure}[!t]
	\captionsetup{justification=centering,margin=2cm,width=\linewidth}
	\centering
	\includegraphics[width=\columnwidth]{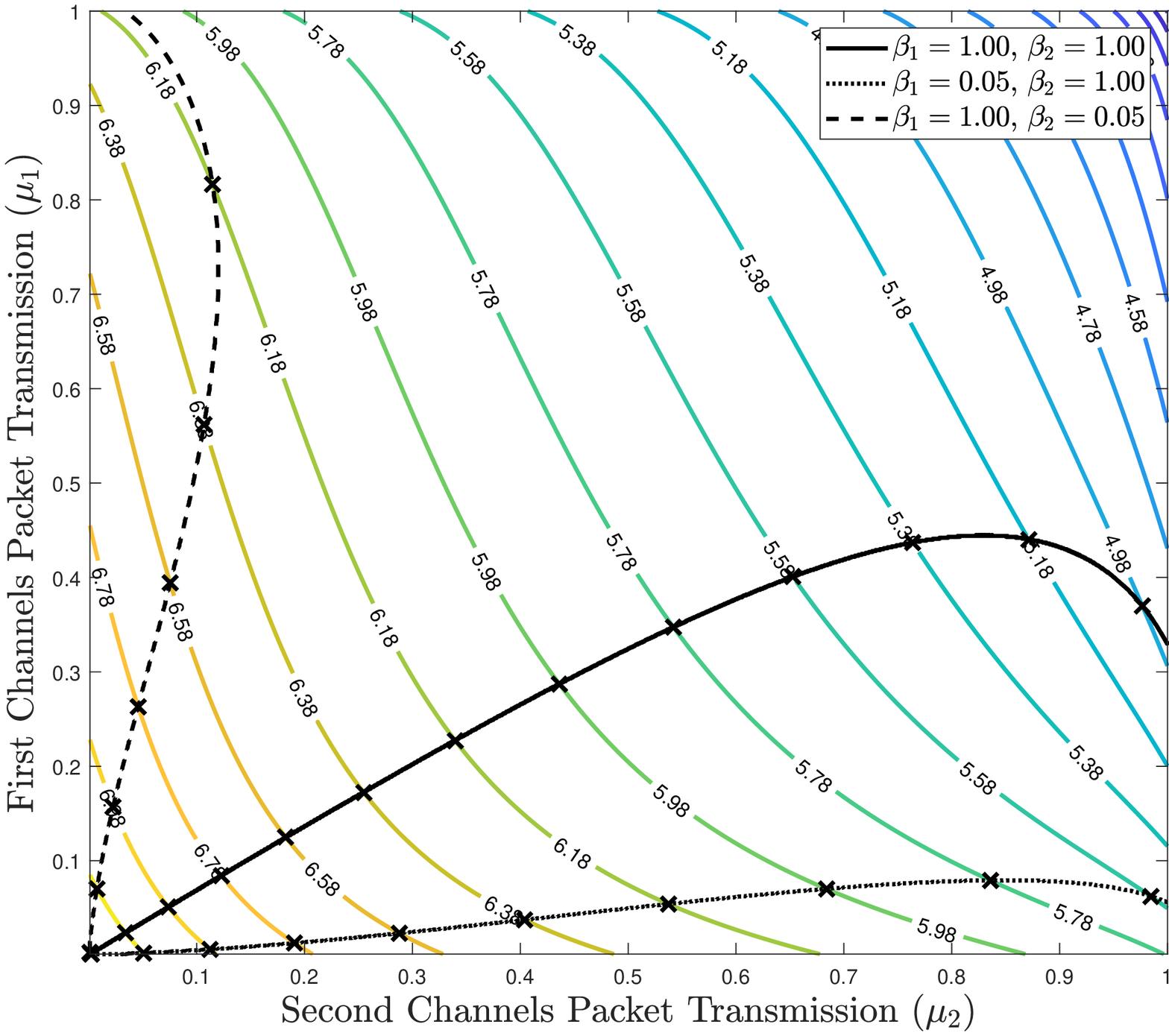}
	\caption[center]{The log expected cost contour plot of the UDP-like protocol for the system~(\ref{eq:mixed-ss}). The mean of the packet transmission variable,~$\Mm$, is varied over the~$\vR^{2}$ region.}\label{fig:UDP-log-cost-mixed-mult}
\end{figure}

In modern communication systems, the probability of packet loss is determined by the performance of multiple processes. %
These processes range from the modulation and coding, which operate in the lower layers, to the routing and flow control, operating in the higher layers. %
While characterising the probability of packet transmission of a modern communication system is challenging in general, the probability of packet transmission decreases monotonically \cite{GallBertDataNetworks} with the resources allocated to the communication system, i.e. bandwidth, power, and delay. %
Indeed, resource allocation is a fundamental problem in communication systems and is often confronted with competing objectives for which the efficiency tradeoffs are difficult to describe analytically. %
In our setting, the extension to the multidimensional actuation channel poses a central question that concerns the design of the communication system of the control system, namely the optimal allocation of communication resources to each of the actuation channel dimensions. %
In the following, we capitalise on the analytical framework developed above and provide a resource allocation framework that optimises the packet loss probabilities for each actuation dimension while satisfying a total budget constraint. %

The set of packet loss probability matrices that achieve a given cost $\alpha\in\mathds{R}_+$ is the set of channel matrices described by %
\beqn\label{eq:alpha_channels}
\Mc_k\left(\! \Ic_k \! \right) =  \left\{ \Mm : J^*\left(\! \Ic_k \! \right) \le \alpha \right\}. 
\eeqn
All the matrices in the set induce a control cost that is upper bounded by $\alpha$ but differ in their use of communication resources, i.e. the packet transmission performance across different dimensions. %
To quantify the use of communication resources in global terms, a communication cost for the system is proposed, defined as %
%is proposed
\beqn\label{eq:channel-cost}
C\left(\Mm\right) \eqdef \sum\limits_{i=1}^{m}\beta_i \mu_i = \tr\left(\betam \Mm\right), 
\eeqn
where~$\betam \in \Sm^{+}_{m\times m }$ is the non-negative definite penalty matrix where each diagonal entry,~$\beta_i \in \R$, is a penalty term corresponding to the cost of communication in that channel. %
The communication cost captures the notion of a total communication budget for the system. %
That being the case, the minimum communication cost is defined as 
\beqn
\label{eq:channel-opt}
C_\alpha\left(\! \Ic_k \! \right) \eqdef \argmin\limits_{\Mm\in\Mc_k\left(\! \Ic_k \! \right)} \ C\left(\Mm\right).
\eeqn
%
%Note that with a small abuse of notation, we express $C_\alpha\left(\! \Ic_k \! \right)$ as a function of $\alpha$ to reflect that the minimum communication cost is determined by the control cost constraint. %
%
The maximum channel efficiency yields a communication setup that minimises the total number of packet losses while maintaining the system control performance. %
In view of this, the communication system configuration that minimises the amount of resources allocated to the actuation channel is %
\beqn
\label{eq:opt_M}
\Mm^* \eqdef \min\limits_{\Mm\in\Mc_k\left(\! \Ic_k \! \right)} \ C\left(\Mm\right).
\eeqn
This cost of communication formulation highlights the importance of Corollary~\ref{cor:com-cost}. %
Specifically, if it is costly to communicate over a particular dimension then allocating as few resources as possible whilst maintaining a desired control cost is desirable, indeed, as few resources as dictated by the minimiser $\Mm^*$. %

The optimisation of the communication channel for the pendulum case study presented in Section \ref{sec:pendulum} is straightforward. %
There is a single communication channel, and therefore, the maximum channel efficiency is the packet transmission value that achieves the optimal control cost of~$\alpha$ with equality. %
However, for the system presented in Section \ref{sec:multi-act-cs}(\ref{eq:mixed-ss}) there are fixed regions of expected cost for both the TCP-like protocol and the UDP-like protocol, as shown in Fig.~\ref{fig:TCP-log-cost-mixed-mult} and Fig.~\ref{fig:UDP-log-cost-mixed-mult}, respectively. %
%
%To achieve the maximal channel efficiency the value of~$J_C\left(\Mm\right)$ is to be minimised for a given expected cost value,~$\alpha$, this is equivalent to finding the point on a contour line that is closest to the origin. %
%
The black dashed lines plotted in Fig.~\ref{fig:TCP-log-cost-mixed-mult} and Fig.~\ref{fig:UDP-log-cost-mixed-mult} correspond to the channel matrices defined in (\ref{eq:opt_M}) that achieve the maximum channel efficiency. %
The values of~$\betav$ are selected for three different cases. %
For the first case the cost is symmetric in both channel dimensions,  i.e. $\betav$ is set to ${\bf I}$ and for the second and third cases the entries are set to $(0.05, 1)$ and $(1,0.05)$, respectively. %
These cases model the situation in which communication across one of the channel dimensions induces larger cost.
Additionally, we also mark the optimal allocation points for each contour line with a black cross. %
%
% on Fig.~\ref{fig:TCP-log-cost-mixed-mult} and Fig.~\ref{fig:UDP-log-cost-mixed-mult} for the TCP-like and the UDP-like protocol respectively. %
%
The same weighting matrices are used for both the UDP-like and the TCP-like figures. %
Note that for the point at which the minimum communication cost is achieved by a matrix $\Mm^*$ with either $\mu_1=1$ or $\mu_2=1$, the optimisation for the channel reduces to the single dimensional case, i.e. one of the channel dimensions is perfect and the other dimension incurs in all packet loss. %
%
%Namely, finding the minimal value of the remaining transmission variable that achieves the~$\alpha$ required with equality. %
%
Furthermore, Fig.~\ref{fig:UDP-log-cost-mixed-mult} shows that for all the three weightings of the channel cost considered, the maximum channel efficiency for $\alpha\ge4.78$ requires at least one of the channels to be perfect. %

The packet loss allocation presented above is not possible with the results~\cite{1}, due to the scalar channel model. %
\section{Conclusion}\label{sec:Conclusion}

%We have studied control systems that experience packet loss on a multidimensional actuation communication channel. %
%
This paper extends existing results in the literature from a scalar channel to the multidimensional channel and, in doing so, derive a proof of the optimal control law for the TCP-like and UDP-like protocols. %
Additionally, an analytic proof is provided showing that the UDP-like LQG expected cost is strictly greater than the TCP-like expected cost. %
It is also shown that the cost is monotonically decreasing in the packet transmission variable. %
The proofs provided are also valid for non-stationary sequences of packet loss. %
This extension is shown by minor adjustments in the conditions. %
%
%Equipped with this result we have also provided a perturbation result that shows the cost induced by the channel is monotonically decreasing in~$\Mm$. This perturbation result provides the necessary tools to pose an optimisation problem, such that the channel operates with a power constraint. In doing so, the system distributes the available power over the~$m$ actuation channels in addition to minimising the cost. We conjecture that in doing so this becomes a water filling problem for the operator. 
%
The maximal cost difference is shown to exist within the~$\vR$ region. %, and in doing so, characterises this region. %
As seen in Fig.~\ref{fig:cost-diff-mixed}, this maximal cost difference exists in the higher dimensional communication channel. %
In showing that the TCP-like protocol outperforms the UDP-like protocol, in terms of LQG cost for all values of~${\bf 0 \prec \Mm \prec I}~$ it has highlighted the problem of deciding which protocol to use. %
The discussion of channel cost optimisation began in Section~\ref{sec:channel-disc}, however, implementing the TCP-like protocol on a system requires the existence of a perfect acknowledgement link. %
The cost of introducing this communication link is not included in the communication cost design. %
The trade-off between choosing the UDP-like or the TCP-like protocol for a particular system must take this additional cost into account.
%
%in addition to the LQG cost difference characterised in this paper.

% that~$\Mm$ determines the expected trajectory of each individual state. This highlights that in order to define the stability of this system the bounds presented in~\cite{1} require altering as a result of changing the channel model. For future work, it is conjectured that the critical threshold for stability presented in~\cite{1} is a sufficient condition for stability of a system. A characterisation of the stability region for the proposed multiple input setting remains an open problem, as does the extension to the imperfect state information case.
%%
% Additionally, stability in this framework needs to be altered to account for the stable and unstable eigenvalues within~$\Am~$. For example, a diagonal~$\Am~$ means any stable states could have~$\Mm_{i,i}\rightarrow0$, this is below the stability threshold presented in~\cite{1}. Additionally, this work assumes that the sensory channel has perfect communication a natural extension could be in dropping this assumption, in doing so it would be a further extension of~\cite{1}. 
%%
%
%

\input{Appendix}

\bibliographystyle{IEEEtran}
\bibliography{References}

\end{document}

%% file: macros.tex
\long\def\comment#1{}

\newfont{\bbb}{msbm10 scaled 700}

\newfont{\bb}{msbm10 scaled 1100}

\newcommand{\EE}{\mbox{\bb E}}

\newcommand{\xv}{{\bf x}}

\renewcommand{\Am}{{\bf A}}
\newcommand{\Bm}{{\bf B}}
\newcommand{\Cm}{{\bf C}}
\newcommand{\Dm}{{\bf D}}
\newcommand{\Em}{{\bf E}}
\newcommand{\Fm}{{\bf F}}
\newcommand{\Gm}{{\bf G}}
\newcommand{\Hm}{{\bf H}}

\newcommand{\Km}{{\bf K}}
\newcommand{\Lm}{{\bf L}}
\newcommand{\Mm}{{\bf M}}

\newcommand{\Qm}{{\bf Q}}

\newcommand{\Sm}{{\bf S}}
\newcommand{\Tm}{{\bf T}}

\newcommand{\Vm}{{\bf V}}

\newcommand{\Bc}{{\cal B}}

\newcommand{\Fc}{{\cal F}}
\newcommand{\Gc}{{\cal G}}

\newcommand{\Ic}{{\cal I}}

\newcommand{\Mc}{{\cal M}}

\newcommand{\Vc}{{\cal V}}
\newcommand{\Xc}{{\cal X}}

\newcommand{\betav}{\hbox{\boldmath$\beta$}}

\renewcommand{\det}{{\hbox{det}}}

\newcommand{\eqdef}{\stackrel{\Delta}{=}}

%% file: Appendix.tex
\appendix

\begin{Applemma}{\ref{lem:error-proof}}
		 
Consider the system modelled by (\ref{eq:plant}) with access to (\ref{eq:information}). Then the following holds
\begin{subequations}
	\beqn
	\EE \left[ \Em_{k}^{\sf T} \Omega{\Em_{k}} {\Big |}\Fc_k\right]\nquad &=&\nquad  \ \tr\left(\Omega_{l} \Sigma_{\pazocal{W}} \right),\\
	\hfill	\EE \left[ \Em_{k}^{\sf T} \Omega{\Em_{k}} {\Big |}\Gc_k\right]\nquad  &=&\nquad  {\pazocal{U}_k \!\! \left(\!\Gc_k\!\right)\!^{\sf T}} \Vsb \left({\bf I} \! \odot \! \Omega_{g} \! \right) \!\! \left( \! {\bf I} - \!\! \Vsb \! \right){\pazocal{U}_k \!\! \left(\!\Gc_k\!\right)\!} \nonumber \\
	&&\qquad\qquad\qquad\!\! +\tr\left(\Omega_{l} \Sigma_{\pazocal{W}}\right),  
	\eeqn 
\end{subequations}
where $\Omega_{l}=\Lambda\Omega\Lambda$.

The proof is split into two parts, one for TCP-like protocol and one for the UDP-like protocol, respectively.

\end{Applemma}
\begin{proof}
\subsubsection{TCP-like protocol}
The expected error in TCP-like follows from (\ref{eq:TCP-prediction-error}). Substituting this into the left-hand side of (\ref{eq:TCP-err}) yields
\beqn
\EE \left[ \Em_{k}^{\sf T} \Omega{\Em_{k}} {\Big |}\Fc_k\right]&=&
%\EE \left[ {\pazocal{W}}_k^{\sf T}\Lambda\Omega\Lambda {\pazocal{W}}_k {\Big |}\Fc_k\right], \nonumber\\
\EE\left[ {\pazocal{W}}_k^{\sf T}\Omega_{l} {\pazocal{W}}_k {\Big |}\Fc_k\right],\label{eq:lem-TCP-err}\\
%	\eeqn 
%
%\noindent where $\Omega_{l}=\Lambda\Omega\Lambda$. The term inside the expectation of (\ref{eq:lem-TCP-err}) is a scalar and therefore the trace of this object is equal to itself. Additionally, expectation and trace are both linear operators, in view of this it can be seen that
%	
%	\beqn
%	\EE \left[ \Em_{k}^{\sf T} \Omega{\Em_{k}} {\Big |}\Fc_k\right]	&=& \EE\left[\mbox{\textnormal tr}\left( {\pazocal{W}}_k^{\sf T}\Omega_{l} {\pazocal{W}}_k\right) {\Big |}\Fc_k\right],\nonumber \\
%&=& \tr\left(\Omega_{l} \EE\left[ {\pazocal{W}}_k {\pazocal{W}}_k^{\sf T}{\Big |}\Fc_k\right]\right) , \\
&=&\tr\left(\Omega_{l}\Sigma_{\pazocal{W}}\right).\label{ap2}
\eeqn 
This completes the TCP-like part of the proof.
\subsubsection{UDP-like protocol}
The error in UDP-like estimation follows from (\ref{eq:UDP-prediction-error}).
%	
%	\beqn%
%	\Em_{k|\Gc_k} &=&\Gamma\left(\Vs - \Vsb\right){\Upsilon_{k|\Ic_k}} + \Lambda {\pazocal{W}}_k.\nonumber
%	\eeqn 
%
Substituting this into the left-hand side of (\ref{eq:UDP-err}) gives
\beqn
&&\nqquad \EE \left[ {\Em^{^{\sf T}}_{k}} \Omega{\Em_{k}} {\Big |}\Gc_k\right]=\EE\left[{\pazocal{U}_k}^{\sf T}\left(\Vs - \Vsb\right)\Omega_{g}\left(\Vs - \Vsb\right){\pazocal{U}_k}{\Big |}\Gc_k\right]  \nonumber \\
&& \qquad \qquad \qquad \qquad \qquad + \EE\left[ {\pazocal{W}}_k^{\sf T}\Lambda\Omega{\Lambda {\pazocal{W}}_k} {\Big |}\Gc_k\right], \nonumber 
\eeqn 
\noindent where we use the fact that ${\pazocal{W}}_k $ is zero mean to eliminate the cross terms. Note that the second term is identical to the TCP-like case, as seen in~(\ref{eq:lem-TCP-err}). Therefore,
\beqn
&&\nqquad \EE \left[ \Em_{k}^{\sf T}\Omega{\Em_{k}}{\Big |}\Gc_k\right] =\EE\left[{\pazocal{U}_k}^{\sf T}\Vs\Omega_{g}{\Vs}{\pazocal{U}_k}{\Big |}\Gc_k\right] \nonumber\\
%	\EE\left[ {\pazocal{W}}_k^{\sf T}\Lambda\Omega{\Lambda {\pazocal{W}}_k} {\Big |}\Gc_k\right]	\nonumber \\
%	&&+\EE\left[{\pazocal{U}_k}^{\sf T}\left(\Vs - \Vsb\right)\Omega_{g}\left(\Vs - \Vsb\right){\pazocal{U}_k}{\Big |}\Gc_k\right],  \nonumber \\
%
%	&&=\EE\left[{\pazocal{U}_k}^{\sf T}\left(\Vs - \Vsb\right)\Omega_{g}\left(\Vs - \Vsb\right){\pazocal{U}_k}{\Big |}\Gc_k\right] +\tr\left(\Omega_{l}\Sigma_{\pazocal{W}}\right), \nonumber\\ 
&&\qquad\qquad\qquad  -{\pazocal{U}_k \!\! \left(\!\Gc_k\!\right)\!^{\sf T}}\Vsb\Omega_{g}\Vsb{\pazocal{U}_k \!\! \left(\!\Gc_k\!\right)\!}+ \tr\left(\Omega_{l}\Sigma_{\pazocal{W}}\right).\nonumber 
\eeqn 
\noindent It follows from Lemma~\ref{lem:quad-expectation-eval} in the Appendix that
\beqn
\EE \!\left[\!\left. \Em_{k}^{\sf T}\! \Omega{\Em_{k}}\right|\Gc_k \! \right]	\! \!\!
%	&&= \Upsilon_{\Gc_k}^{\sf T}\Vsb\Omega_{g}\Vsb{\pazocal{U}_k \!\! \left(\!\Gc_k\!\right)\!}+\Upsilon_{\Gc_k}^{\sf T}\Vsb\Omega_{g}\left({\bf I} \odot \Vsb\right){\pazocal{U}_k \!\! \left(\!\Gc_k\!\right)\!}    \nonumber \\
%	&& \qquad \qquad-\Upsilon_{\Gc_k}^{\sf T}\Vsb\Omega_{g}\Vsb{\pazocal{U}_k \!\! \left(\!\Gc_k\!\right)\!} + \tr\left(\Omega_{l}\Sigma_{\pazocal{W}}\right), \nonumber \\
= \pazocal{U}_k \!\! \left(\!\Gc_k\!\right)\!^{\sf T} \! \Vsb \!\! \left({\bf I}\odot \Omega_{g} \! \right) \!\!\! \left({\bf I}- \Vsb\right) \!\! {\pazocal{U}_k \!\! \left(\!\Gc_k\!\right)\!}  + \tr\left(\Omega_{l}\Sigma_{\pazocal{W}} \! \right) \! \! .\nonumber 
\eeqn 
This concludes the proof.
\end{proof}

\begin{lemma}\label{lem:quad-expectation-eval}
	It is proved that:
	\beqn
	&&\nqquad\nquad\EE\left[{\pazocal{U}_k}^{\sf T}\Vs^{^{\sf T}}\Omega_{g}\Vs{\pazocal{U}_k}{\Big |}\Gc_k \right] = {\pazocal{U}_k \!\! \left(\!\Gc_k\!\right)\!}^{\sf T}\Vsb\Omega_{g}\Vsb{\pazocal{U}_k \!\! \left(\!\Gc_k\!\right)\!} \nonumber \\
	&&  \qquad \qquad+ {\pazocal{U}_k \!\! \left(\!\Gc_k\!\right)\!}^{\sf T} \Vsb \left({\bf I} \! \odot \! \Omega_{g} \! \right) \!\! \left( \! {\bf I} - \!\! \Vsb \! \right){\pazocal{U}_k \!\! \left(\!\Gc_k\!\right)\!}, \label{eq:appendix}
	\eeqn 
where ${\bf I}$ is the identity matrix and $\odot$ is the element wise Hadamard product.
\end{lemma}
\begin{proof}
The left hand side of (\ref{eq:appendix}) is scalar, and therefore
\beqn
&& \nqquad \!\! \EE\!\! \left[ \pazocal{U}_k \!\! \left(\!\Gc_k\!\right)\!^{\sf T} \Vs^{\sf T} \Omega_{g}\Vs{ \pazocal{U}_k \!\! \left(\!\Gc_k\!\right)\!^{\sf T}}{\Big|}\Gc_k\right]\nonumber \\
%
%	&& = \EE\left[U_0{\bf V}_0\Omega_{g}_{(1,1)}}{\bf V}_0 U_0 +U_0 {\bf V}_0\Omega_{g}_{(1,2)}}{\bf V}_{1}U_{1} + \right.\nonumber \\
%	&&\qquad \qquad \qquad \qquad \left. \dots +U_{Nm-1}{\bf V}_{Nm-1}\Omega_{g}_{(Nm,Nm)}}{\bf V}_{Nm-1}U_{Nm-1}\right] \nonumber \\
%
&& \nqquad \nquad \ = \! \EE \!\! \left[ \! U_0 \!\! \left(\!\Gc_k\!\right)\!\! {\bf V}_0\Omega_{g_{(1,1)}} \! \!\!\! {\bf V}_0 U_0  \!\! \left(\!\Gc_k\!\right)\!\! {\Big|} \! \Gc_k \!\! \right] \!\! + \! \EE \!\! \left[ \! U_0  \!\! \left(\!\Gc_k\!\right)\! {\bf V}_0\Omega_{g_{(1,2)}}\! \!\!\! {\bf V}_1 U_1  \!\! \left(\!\Gc_k\!\right)\!\! {\Big|} \! \Gc_k \!\! \right] \nonumber \\
&&  \nqquad\dots\  + \EE\left[U_{Nm-1} \!\! \left(\!\Gc_k\!\right)\!\! {\bf V}_{Nm-1}\Omega_{g_{(Nm,Nm)}}{\bf V}_{Nm-1}U_{Nm-1}  \!\! \left(\!\Gc_k\!\right)\!\! {\Big|}\Gc_k\right] ,\nonumber \\
%
%&=& \nquad \sum\limits_{i=1}^{Nm}\left(\EE\left[U_{i-1}{\bf V}_{i-1}\Omega_{g_{(i,i)}}{\bf V}_{i-1}U_{i-1}{\Big|}\Gc_k\right] \vphantom{\sum^{Nm}_{i}}\right.\nonumber \\
%&&  \left. + \sum\limits_{j=1,j\neq i}^{Nm}\EE\left[U_{i-1}{\bf V}_{i-1}\Omega_{g_{(i,j)}}{\bf V}_{j-1}U_{j-1}{\Big|}\Gc_k\right] \right),\nonumber \\
%
%&=& \sum\limits_{i=1}^{Nm}\left(\EE\left[U_{i-1}{\bf V}_{i-1}\Omega_{g_{(i,i)}}U_{i-1}{\Big|}\Gc_k\right] \vphantom{\sum^{Nm}_{i}}\right. \nonumber \\
%&& \qquad \left.+ \sum\limits_{j=1,j\neq i}^{Nm}\EE\left[U_{i-1}{\bf V}_{i-1}\Omega_{g_{(i,j)}}{\bf V}_{j-1}U_{j-1}{\Big|}\Gc_k\right] \right),\nonumber \\ 
%
&=& \nquad \sum\limits_{i=1}^{Nm}\left(  U_{i-1}\!\! \left(\!\Gc_k\!\right)\! \Mm_{i-1}\Omega_{g_{(i,i)}} U_{i-1}\!\! \left(\!\Gc_k\!\right)\!\!\vphantom{\sum\limits_{j=1,j\neq i}^{Nm}}\right.\nonumber \\
&&\left. + \sum\limits_{j=1,j\neq i}^{Nm}  U_{i-1}\!\! \left(\!\Gc_k\!\right)\! \Mm_{i-1} \Omega_{g_{(i,j)}}\Mm_{j-1}  U_{j-1}\!\! \left(\!\Gc_k\!\right)\!\!,\right)\nonumber \\
%&=& \sum\limits_{i=1}^{Nm} \sum\limits_{j=1,j\neq i}^{Nm}U_{i-1}\Mm\Omega_{g_{(i,i)}}U_{i-1} +U_{i-1}\Mm\Omega_{g_{(i,j)}}U\Mm_{j-1}, \nonumber \\
%
%	&=& \sum\limits_{i=1}^{Nm} \sum\limits_{j=1}^{Nm}U_{i-1}\Mm\Omega_{g_{(i,i)}}U_{i-1} +U_{i-1}\Mm\Omega_{g_{(i,j)}}\MmU_{j-1} \nonumber \\
%	&& \qquad \qquad\qquad \qquad\qquad \qquad \qquad\quad  -U_{i-1}\Mm\Omega_{g_{(i,i)}}\MmU_{i-1}\qquad\nonumber \\
%
%
&=& \nquad \sum\limits_{i=1}^{Nm} \sum\limits_{j=1}^{Nm}  U_{i-1}\!\! \left(\!\Gc_k\!\right)\!\!\Mm_{i-1}\Omega_{g_{(i,i)}}\left(1-\Mm_{i-1}\right) U_{i-1}\!\! \left(\!\Gc_k\!\right)\!\! \nonumber \\
&& \qquad +  U_{i-1}\!\! \left(\!\Gc_k\!\right)\!\!\Mm_{i-1} \Omega_{g_{(i,j)}}\Mm_{j-1}  U_{j-1}\!\! \left(\!\Gc_k\!\right)\!\!, \nonumber\\
&=& \nquad {\pazocal{U}_k \!\! \left(\!\Gc_k\!\right)\!}^{\sf T}\Vsb \! \left({\bf I} \! \odot \! \Omega_{g} \! \right) \!\! \left( \! {\bf I} - \!\! \Vsb \! \right){\pazocal{U}_k \!\! \left(\!\Gc_k\!\right)\!} + {\pazocal{U}_k \!\! \left(\!\Gc_k\!\right)\!}^{\sf T}\Vsb\Omega_{g} \!\! \Vsb{\pazocal{U}_k \!\! \left(\!\Gc_k\!\right)\!} ,\nonumber
\eeqn 	
where bracketed subscripts represent the $(i,j)$-th element. %
This concludes the proof. %
\end{proof}

\begin{Appthm}{\ref{th:contol}}
Consider the closed-loop systems shown in Fig.~\ref{fig:TCP-like-diagram} and Fig.~\ref{fig:UDP-like-diagram}, with plant dynamics given in~(\ref{eq:plant}), protocol dependent information sets given in~(\ref{eq:information}) and controller cost function given in~(\ref{eq:opt-lqg}). Then the optimal cost for the TCP-like protocol is
%
%	Consider a Gauss-Markov system described by~(\ref{eq:plant}) that is experiencing actuation packet losses. The system operator implements either a TCP-like or UDP-like protocol, depicted in Fig. \ref{fig:UDP-like-diagram} and Fig. \ref{fig:TCP-like-diagram}. The optimal cost for the TCP-like protocol is
%
\beqn
J^*\left(\Fc_k\right)=\!\! X_k^{\sf T}\!\!\left( \Qm +\Omega_{p} \right)\!\!X_k\!+\! \tr\left(\Sigma_{\pazocal{W}}\Omega_{l}\!\right)\!\!-\!  X_k^{\sf T}\! \Fm^{\sf T}\!\!\Gm^{\sminus1}\!\!\left(\!\Fc_k\!\right)\!\!\Vsb \Fm X_k ,\nonumber
\eeqn
and the optimal cost for the UDP-like protocol is
\beqn
J^*\left(\Gc_k\right)=\!\! X_k^{\sf T}\!\!\left( \Qm +\Omega_{p} \right)\!\!X_k\!+\! \tr\left(\Sigma_{\pazocal{W}}\Omega_{l}\!\right)\!\!-\!  X_k^{\sf T}\! \Fm^{\sf T}\!\!\Gm^{\sminus1}\!\!\left(\!\Gc_k\!\right)\!\!\Vsb \Fm X_k.\nonumber
\eeqn

As with Lemma~\ref{lem:error-proof} the proof is split into two parts to account for the two protocols.
\end{Appthm}
\begin{proof}
	\setcounter{subsubsection}{0}
	\subsubsection{Optimal Cost for the TCP-like Protocol}
	\noindent Substituting (\ref{eq:TCP-err}) into (\ref{eq:opt-lqg-simplest}), noting that under the TCP-like protocol the error term does not depend on $\pazocal{U}_k$, gives
	\beqn
	\label{eq:plant7}
 J^*\!\! \left( \! \Fc_k \! \right) \nquad  &=& \nquad  X_k^{\sf T} \! \left(\Qm + \Omega_{p}\right)X_k + \tr\left(\Omega_{l}\Sigma_{\pazocal{W}}\right)\nonumber \\
	%
	%	&&+ \min_{\pazocal{U}_k}\left\{\tr\left(\Omega_{l}\Sigma_{\pazocal{W}}\right)+ {\pazocal{U}_k \!\! \left(\!\Fc_k\!\right)\!^{\sf T}} \Vsb \left(2\Fm X_k + \left(\Omega_{g} \Vsb+ \Psi \right){\pazocal{U}_k \!\! \left(\!\Fc_k\!\right)\!}\right)\right\},\nonumber  \\
	%%
	%	&& =  X_k^{\sf T} \left(\Qm + \Omega_{p}\right)X_k+ \tr\left(\Omega_{l}\Sigma_{\pazocal{W}}\right) \nonumber \\
	%	&& \quad +\min_{\pazocal{U}_k}\left\{  {\pazocal{U}_k}^{\sf T} \Vsb \left(2\Fm X_k + \left(\Omega_{g} \Vsb+ \Psi \right){\pazocal{U}_k}\right)\right\}\nonumber  \\
	%%
	%	&&  \nonumber \\
	&& \nqquad + \min_{\pazocal{U}_k \!\! \left(\!\Gc_k\!\right)}\!\left\{\! \pazocal{U}_k \!\! \left(\!\Gc_k\!\right)\!^{\sf T}\Vsb \!\! \left(2\Fm X_k \!+\!\left(\Omega_{g}\Vsb\!+\!\Psi \right){\pazocal{U}_k \!\! \left(\!\Gc_k\!\right)\!}\right) \right\}\!\!. \label{eq:TCP-cost-minimisation}
	\eeqn 
Note that $\left(\Omega_{g} \Vsb +\Psi\right)$ is positive definite, and therefore,~(\ref{eq:TCP-cost-minimisation}) is convex. Taking the derivative of the cost with respect to ${\pazocal{U}_k}$ yields
	\beqn
	{{\partial 	J^*\left(\Fc_k\right)}  \over {\partial {\pazocal{U}_k}}} = 2\Vsb\left(\Fm X_k + \left(\Omega_{g} \Vsb +\Psi\right){\pazocal{U}_k \!\! \left(\!\Fc_k\!\right)\!} \right).\label{eq:information0}
	\eeqn
	\noindent Solving for all $\Vsb \neq {\bf 0}$, the minimising value of~${\pazocal{U}_k \!\! \left(\!\Fc_k\!\right)\!} $ is found to be
	\newsavebox{\smlAun}% Box to store smallmatrix content
	\savebox{\smlAun}{$\left(\begin{smallmatrix} 1.03&0.005\\0.35&1.01\end{smallmatrix}\right)$}
	\beqn
	\pazocal{U}_k^*\!\!\left(\!\Fc_k\!\right)\! \eqdef - \left(\Omega_{g} \Vsb +\Psi\right)^{\sminus1} \Fm X_k .
	\eeqn
	Denoting $\left(\Omega_{g} \Vsb +\Psi\right)$ by $\Gm^{\sminus1}\!\!\left(\!\Fc_k\!\right)\!$ and substituting ${\Upsilon_{k|\mathcal {F}_k}^{*}}$ into (\ref{eq:TCP-cost-minimisation}) results in the optimal expected cost for the operator, described by
	\beqn
	J^*\left(\Fc_k\right)=\!\! X_k^{\sf T}\!\!\left( \Qm +\Omega_{p} \right)\!\!X_k\!+\! \tr\left(\Sigma_{\pazocal{W}}\Omega_{l}\!\right)\!\!-\!  X_k^{\sf T}\! \Fm^{\sf T}\!\!\Gm^{\sminus1}\!\!\left(\!\Fc_k\!\right)\!\!\Vsb \Fm X_k .\nonumber
	\eeqn
	This concludes the TCP-like part of the proof.
	%\begin{figure}[!t]
	%	\captionsetup{justification=centering,margin=2cm,width=\linewidth}
	%	%	\includegraphics[width=\linewidth]{Figures/Cost_with_Nu-eps-converted-to.pdf}
	%	\begin{subfigure}{\linewidth}
	%		\hspace{-1.1cm}\includegraphics[width=1.2\linewidth]{Figures/TCP_Cost_Mat_Nu.eps}
	%	\end{subfigure}
	%	\begin{subfigure}{\linewidth}
	%		\hspace{-1.1cm}\includegraphics[width=1.2\linewidth]{Figures/UDP_Cost_Mat_Nu.eps}
	%	\end{subfigure}
	%	\begin{subfigure}{\linewidth}
	%		\hspace{-1.1cm}\includegraphics[width=1.2\linewidth]{Figures/TCP_Cost_Mat_Nu_1.eps}
	%	\end{subfigure}
	%	\begin{subfigure}{\linewidth}
	%		\hspace{-1.1cm}\includegraphics[width=1.2\linewidth]{Figures/UDP_Cost_Mat_Nu_1.eps}
	%	\end{subfigure}
	%	%	\includegraphics[width=\linewidth]{Cost_with_Nu.eps}
	%	\caption{The cost for each protocol when varied in $\Mm$. For a matrix $\Mm$ i.e $\Mm_{\{i,i\}}=\mu_i$, with dynamics matrix ${\bf A}=$~\usebox{\smlA}, and control matrix ${\bf B}= {\bf I}$. When $\mu_i$ is varied $\mu_j $ is fixed at $\mu_j=0$, where $ j\neq i$, for the top 2 plots and at $\mu_j=1$ for the bottom two.}\label{fig3}
	%\end{figure}
	%
	%\vspace{-1cm}
	\subsubsection{Optimal Cost for the UDP-like Protocol}
	\noindent Combining (\ref{eq:UDP-err}) and (\ref{eq:opt-lqg-simplest}) the optimal cost function for the UDP-like protocol is given by:
	\beqn
	&&\nquad J^*\left(\Gc_k\right) = X_k^{\sf T}\left( \Qm +\Omega_{p} \right)X_k +\tr\left(\Omega_{l} \Sigma_{\pazocal{W}}\right)  \nonumber\\
	&& \nqquad + \min_{\pazocal{U}_k \!\! \left(\!\Gc_k\!\right)}\!\left\{\! \pazocal{U}_k \!\! \left(\!\Gc_k\!\right)\!^{\sf T} \! \Vsb \!\! \left( \! 2\Fm X_k \! + \! \left( \! \Omega_{g} \!\! \Vsb \!\! + \! \Psi \! + \!\! \left({\bf I} \! \odot \! \Omega_{g} \! \right) \!\! \left( \! {\bf I} - \!\! \Vsb \! \right) \!\! \right) \!\! {\pazocal{U}_k \!\! \left(\!\Gc_k\!\right)\!}\right) \!\! \right\}\!\!.\nonumber
	\eeqn
	Following the same process as with the TCP-like case, noting that $\left(\Omega_{g}\Vsb\!+\!\Psi + \left({\bf I} \! \odot \! \Omega_{g} \! \right) \!\! \left( \! {\bf I} - \!\! \Vsb \! \right)\right)$ is positive definite, and therefore the minimisation is convex yields
	\beqn
	{{\partial 	J^* \! \left( \! \Gc_k \! \right) } \over {\partial {\pazocal{U}_k}}} =  2\Vsb \!\! \left(\Fm X_k \!+\!\left(\Omega_{g}\Vsb\!+\!\Psi + \left({\bf I} \! \odot \! \Omega_{g} \! \right) \!\! \left( \! {\bf I} - \!\! \Vsb \! \right) \! \right){\pazocal{U}_k \!\! \left(\!\Gc_k\!\right)\!}\right),\nonumber
	\eeqn
	%%
	%\noindent The $\left(\Psi +\left({\bf I} -\Vsb\right)\left({\bf I} \odot \Omega_{g}\right) + \Omega_{g}\Vsb\right)$ term contains positive definite and positive semi-definite terms and is therefore strictly positive definite, as a result it is invertible. For all $\Vsb \neq {\bf 0}$ the minimising ${\pazocal{U}_k \!\! \left(\!\Gc_k\!\right)\!} $ is:
	%%
	meaning that the optimal value of~$\pazocal{U}_k^*\!\! \left(\!\Gc_k\!\right)\!$ is
	\beqn
	{\pazocal{U}_k^*\!\! \left(\!\Gc_k\!\right)\!} \eqdef - \left(\Psi +\left({\bf I} \! \odot \! \Omega_{g} \! \right) \!\! \left( \! {\bf I} - \!\! \Vsb \! \right) + \Omega_{g}\Vsb\right)^{\sminus1} \Fm X_k.
	\eeqn
	\noindent Re-labelling $\left(\Psi +\left({\bf I} \! \odot \! \Omega_{g} \! \right) \!\! \left( \! {\bf I} - \!\! \Vsb \! \right) + \Omega_{g}\Vsb\right)$ as $\Gm\!\!\left(\!\Gc_k\!\right)\!$ and substituting ${\pazocal{U}_k^*\!\! \left(\!\Gc_k\!\right)\!}$ into (\ref{eq:TCP-cost-minimisation}) yields the optimal expected cost for the operator:
	\beqn\label{eq:UDP-opt}
	J^*\left(\Gc_k\right)=\!\! X_k^{\sf T}\!\!\left( \Qm +\Omega_{p} \right)\!\!X_k\!+\! \tr\left(\Sigma_{\pazocal{W}}\Omega_{l}\!\right)\!\!-\!  X_k^{\sf T}\! \Fm^{\sf T}\!\!\Gm^{\sminus1}\!\!\left(\!\Gc_k\!\right)\!\!\Vsb \Fm X_k.
	\eeqn
	This concludes the proof.
\end{proof}

\begin{Applemma}{\ref{lem:cost-diff-deriv}}
The cost difference between the UDP-like and the TCP-like protocols is defined as
\beqn
J_\Delta^*\left(\Vsb\right) \eqdef J^*\left(\Gc_k\right)- J^*\left(\Fc_k\right)&>&0.
\eeqn
The derivative of this cost difference is
\beqn
{{\partial J_\Delta^*\left(\Vsb\right) }\over{\partial \Vsb }}	\nquad 
&=&\nquad   X_k^{\sf T}\Fm^{\sf T} \left(
\Gm^{\sminus1}\!\!\left(\!\Gc_k\!\right)\!\left(\vphantom{\Big|} \left(1-2\Vsb \right) \Omega_{d} \right.\right.\nonumber\\
&&\nqquad\nqquad\nqquad \ \  \left.\left.  - \Vsb \!\! \left( \!\! 1 \sminus \!\! \Vsb \! \right) \!\!\! \left[ \! \Omega_{h}\Gm^{\sminus1}\!\!\left(\!\Gc_k\!\right)\!\Omega_{d} + \Omega_{d} \Gm^{\sminus1}\!\!\left(\!\Fc_k\!\right)\! \Omega_{g}\right] \! \right) \!\! \Gm^{\sminus1}\!\!\left(\!\Fc_k\!\right)\! \right) \!\! \Fm X_k \! , \quad \ 
\eeqn
where $\Omega_{d}=\left({\bf I} \odot \Omega_{g}\right)$ and $\Vsb \in \vR$.
\end{Applemma}
\begin{proof}
	In order to proceed with the proof we analyse the matrices $\Gm^{\sminus1}\!\!\left(\!\Fc_k\!\right)\!$ and $\Gm^{\sminus1}\!\!\left(\!\Gc_k\!\right)\!$. %
	We define the mappings $\Gm_{\Fm}\left(\alpha\right):\R \rightarrow \R^{n\times n}$ and $\Gm_{\Gm}\left(\alpha\right):\R \rightarrow \R^{n\times n}$ as
	\begin{subequations}\label{eq:G-mappings}
		\beqn
		\Gm_{\Fm}\left(\alpha\right)&=& \left(\alpha\Omega_{g} +\Psi\right)^{\sminus1}, \\
		\Gm_{\Gm}\left(\alpha\right)&=& \left(\alpha\Omega_{h} + \Omega_{d} + \Psi\right)^{\sminus1}, 
		\eeqn
	\end{subequations}
	where $\alpha\in\vR$, $\Gm_{\Fm} \!\! \left( \! \Vsb \! \right) \!\! =\Gm^{\sminus1}\!\!\left(\!\Fc_k\!\right)\!$, $\! \Gm_{\Gm} \!\!\! \left( \! \Vsb \! \right) \!\! =\Gm^{\sminus1}\!\!\left(\!\Gc_k\!\right)\!$, $\Omega_{d}= \left({\bf I} \odot \Omega_{g}\right) $, and $\Omega_{h} = \left(\Omega_{g} - \left({\bf I} \odot \Omega_{g}\right) \right)$. %
	It should be noted that these two functions share the same inherent structure. Specifically, both of these functions are specialisations of the more general function
	\beqn
	\Gm_{\bf I}\left(\alpha,\Am,\Bm\right)= \left(\alpha\Am +\Bm\right)^{\sminus1}, 
	\eeqn
	where $\Am\in\R^{nN\times nN}$, $\Bm \in \R^{nN \times nN}$, and $\Gm_{\bf I}\left(\alpha,\Am,\Bm\right)$ is defined as the mapping $\Gm_{\bf I}\left(\alpha,\Am,\Bm\right):\R\times\R^{nN\times nN}\times\R^{nN\times nN} \rightarrow \R^{nN\times nN}$. %
	This results in a single function that is equivalent to~(\ref{eq:G-mappings}) where it is noted that $\Gm_{\bf I} \left(\alpha, \Omega_{g}, \Psi\right) = \Gm_{\Fm} \left(\alpha \right) $ and $\Gm_{\bf I} \left(\alpha, \Omega_{h}, \left[\Omega_{d}+\Psi\right] \right) = \Gm_{\Gm} \left(\alpha \right)$. %
	Additionally, all inputs of the function $\Gm_{\bf I}\left(\alpha,\Am,\Bm\right)$ are symmetric and $\Bm$ is a diagonal positive definite matrix for both of the protocols considered. %
	From this point, due to the fact that the matrices $\Am$ and $\Bm$ are constants, we simplify the notation to $\Gm_{\bf I}\left(\alpha\right)$. %
	Note that the results below apply to both the UDP-like and the TCP-like protocols. The function $\Gm_{\bf I}\left(\alpha \right)$ has a first derivative given by
	\beqn
	{{\partial}\over{\partial \beta }}\Gm_{\bf I}\left(\beta\right)\nquad \ &=&\nquad \ {{\partial}\over{\partial \beta }} \left(\beta\Am +\Bm\right)^{\sminus1} \nonumber \\
	&=&\nquad \  \minus\left(\beta\Am +\Bm\right)^{\sminus1}\!\!\left[{{\partial}\over{\partial \beta }} \left(\beta\Am +\Bm\right)\right]\!\! \left(\beta\Am +\Bm\right)^{\sminus1}\!\!\! ,\ \ \ \ \ \label{eq:inverse-derivitive-sub}
	\eeqn
	where in~(\ref{eq:inverse-derivitive-sub}) the derivative is recast according to~\cite[17.3(a)]{seber}. Executing the derivative operator results in
	\beqn
	{{\partial}\over{\partial \beta }}\Gm_{\bf I}\left(\beta\right)\nquad \
	% &=&\nquad \ 
%	-\left(\beta\Am +\Bm\right)^{\sminus1}\Am \left(\beta\Am +\Bm\right)^{\sminus1} \nonumber \\
	%
	&=&\nquad \ 
	-\Gm_{\bf I}\left(\beta\right)\Am \Gm_{\bf I}\left(\beta\right). \label{eq:G-deriv}
	\eeqn
	With this in mind the derivative of the cost difference is computed to be
	\beqn
	\nquad {{\partial J_\Delta^* \!\! \left( \! \Vsb \! \right) }\over{\partial \Vsb }}	\nquad  
	&=& \nquad \left[ {{\partial}\over{\partial \Vsb }} \Vsb\left(1-\Vsb\right) \tr \left(\! \Gm_{\Gm} \!\!\! \left( \! \Vsb \! \right) \!\! \Omega_{d}\Gm_{\Fm} \!\! \left( \! \Vsb \! \right) \!\! \Lm \right)\right]\!\!,\!\!  \\
	%
%	&=& \nquad \left[ {{\partial}\over{\partial \Vsb }} \Vsb\left(1-\Vsb\right) \right]\tr \left(\! \Gm_{\Gm} \!\!\! \left( \! \Vsb \! \right) \!\! \Omega_{d}\Gm_{\Fm} \!\! \left( \! \Vsb \! \right) \!\! \Lm \right) \nonumber \\
%	&& \! \! \nquad  +  \Vsb\left(1-\Vsb\right) \left[ {{\partial}\over{\partial \Vsb }} \tr \left(\! \Gm_{\Gm} \!\!\! \left( \! \Vsb \! \right) \!\! \Omega_{d}\Gm_{\Fm} \!\! \left( \! \Vsb \! \right) \!\! \Lm \right)\right]\!\!,\!\!  \\
	%
	&=& \nquad \left[ {{\partial}\over{\partial \Vsb }} \Vsb\left(1-\Vsb\right) \right]\tr \left(\! \Gm_{\Gm} \!\!\! \left( \! \Vsb \! \right) \!\! \Omega_{d}\Gm_{\Fm} \!\! \left( \! \Vsb \! \right) \!\! \Lm \right) \nonumber \\
	&& \! \! \nquad  +  \Vsb\left(1-\Vsb\right) \tr \left(\left[ {{\partial \! \Gm_{\Gm} \!\!\! \left( \! \Vsb \! \right) \!\! }\over{\partial \Vsb }} \right] \Omega_{d}\Gm_{\Fm} \!\! \left( \! \Vsb \! \right) \!\! \Lm \right)\!\! \nonumber  \\
	&& \! \! \nquad  +  \Vsb\left(1-\Vsb\right) \tr \left(\! \Gm_{\Gm} \!\!\! \left( \! \Vsb \! \right) \!\!  \Omega_{d} \left[ {{\partial \Gm_{\Fm} \!\! \left( \! \Vsb \! \right) \!\! }\over{\partial \Vsb }}\right]\Lm \right)\!\!,\!\! 
	\eeqn
	where in the third line the property~\cite[17.5 pg.353]{seber} in conjunction with the product rule is utilised. %
	At this stage, implementing the result seen in (\ref{eq:G-deriv}) yields
	\beqn
	\nquad {{\partial J_\Delta^* \!\! \left( \! \Vsb \! \right) }\over{\partial \Vsb }}\nquad 
	%
%	&=&  \nquad \left(1-2\Vsb\right) \tr \left(\! \Gm_{\Gm} \!\!\! \left( \! \Vsb \! \right) \!\! \Omega_{d}\Gm_{\Fm} \!\! \left( \! \Vsb \! \right) \!\! \Lm \right) \nonumber \\
%	%
%	&& \! \! \nquad  -  \Vsb\left(1-\Vsb\right) \tr \left(\! \Gm_{\Gm} \!\!\! \left( \! \Vsb \! \right) \!\! \Omega_{h} \! \Gm_{\Gm} \!\!\! \left( \! \Vsb \! \right) \!\! \Omega_{d}\Gm_{\Fm} \!\! \left( \! \Vsb \! \right) \!\! \Lm \right)\!\! \nonumber  \\
%	%
%	&& \! \! \nquad  -  \Vsb\left(1-\Vsb\right) \tr \left(\! \Gm_{\Gm} \!\!\! \left( \! \Vsb \! \right) \!\!  \Omega_{d} \Gm_{\Fm} \!\! \left( \! \Vsb \! \right) \!\! \Omega_{g} \Gm_{\Fm} \!\! \left( \! \Vsb \! \right) \!\! \Lm \right)\!\!, \nonumber  \\
%	%
%	%
%	&=&\! \! \nquad \  X_k^{\sf T}\Fm^{\sf T} \left[ \left(1-2\Vsb \right) \! \Gm_{\Gm} \!\!\! \left( \! \Vsb \! \right) \!\! ) \Omega_{d}\Gm_{\Fm} \!\! \left( \! \Vsb \! \right) \!\!  \right.\nonumber\\
%	%
%	&&\nqquad \   - \Vsb\left(1-\Vsb\right)\! \Gm_{\Gm} \!\!\! \left( \! \Vsb \! \right) \!\! \Omega_{h}\Gm_{\Gm}\left(\Vsb \right) \Omega_{d}\Gm_{\Fm}\left(\Vsb \right) \!\! \nonumber \\
%	%
%	&&\nqquad \ \left.  - \Vsb\left(1-\Vsb\right)\! \Gm_{\Gm} \!\!\! \left( \! \Vsb \! \right) \!\!  \Omega_{d} \Gm_{\Fm} \!\! \left( \! \Vsb \! \right) \!\!   \Omega_{g}\Gm_{\Fm} \!\! \left( \! \Vsb \! \right) \!\!   \right]\Fm X_k \!, \nonumber\\
	%
	%
	&=&\nquad \  X_k^{\sf T}\Fm^{\sf T} \left[\Gm_{\Gm}\!\!\left(\Vsb\right)\left[ \left(1-2\Vsb \right) \Omega_{d} \right.\right.\nonumber\\
	&&\nqquad \nqquad \nqquad  \left.\left.  - \Vsb \!\! \left( \! 1 \sminus \! \Vsb \! \right) \!\!\! \left[\Omega_{h}\Gm_{\Gm}\!\!\left(\Vsb\right)\Omega_{d} + \Omega_{d} \Gm_{\Fm}\!\!\left(\Vsb\right) \Omega_{g}\right] \right]\Gm_{\Fm}\!\!\left(\Vsb\right) \right]\Fm X_k,
	\eeqn
	which corresponds to (\ref{eq:cost-diff-deriv}). %
	This concludes the proof.
\end{proof}

\begin{Applemma}{\ref{lem:Nu-Sols}}
The relation
\beqn
&&	\det\left(\vphantom{\Big|}\Gm_{\Gm}\!\!\left(\Vsb\right)\left[ \left(1-2\Vsb \right) \Omega_{d} \right.\right. \nonumber \\
&& \nqquad\nqquad \left.\left.
 - \Vsb \!\! \left( \! 1 \sminus \! \Vsb \! \right) \!\!\! \left[ \Omega_{h} \Gm_{\Gm}\!\!\left(\Vsb\right) \Omega_{d} + \Omega_{d} \Gm_{\Fm}\!\!\left(\Vsb\right) \Omega_{g} \right] \right]\Gm_{\Fm}\!\!\left(\Vsb\right)\vphantom{\Big |} \!\!\! \right) = 0 
\eeqn
has $2Nm$ many solutions, specifically,
\begin{subequations}
	\beqn
	\Vsb^D_{2i-1} &=&  {{1}\over{1 + \sqrt{1+\lambda_i}}},\\ 
	\Vsb^D_{2i} &=&  {{1}\over{1 - \sqrt{1+\lambda_i}}},
	\eeqn	
\end{subequations}
where $\Vsb^D_i $ correspondences to the $i$-th solution for (\ref{eq:det-0}) and $\lambda_i$ is the $i$-th eigenvalue of the matrix,
\beqn
\left(\Omega_{g}\Omega_{d}^{\sminus1}\left(\Omega_{g} + \Psi\right) + \Psi \Omega_{d}^{\sminus1}\Omega_{h} \right) \left(\Omega_{g} + \Psi\right)^{\sminus1} \Omega_{d}\Psi^{\sminus1}.
\eeqn
\end{Applemma}
\begin{proof}
	The first thing to establish is that the determinant of a matrix product is equal to the product of the respective matrix determinants~\cite[4.31(a)]{seber}. %
	Specifically, for any two matrices of equal dimensions
	\beqn
	\det\left(\Cm\Dm\right)=\det\left(\Cm\right)\det\left(\Dm\right). \label{eq:det-rule}
	\eeqn
	Therefore, multiplying~(\ref{eq:det-0}) from the left and right by $\det\left(\Omega_{d}^{\sminus1}\right)\det\left(\Gm_{\Gm}^{\sminus1}\left(\Vsb\right)\right)$ and $\det\left(\Gm_{\Fm}^{\sminus1}\left(\Vsb\right)\right)\det\left(\Omega_{d}^{\sminus1}\right)$ respectively gives
	\beqn
	\det \!\! \left( \!\!\vphantom{\Big|} \left[ \!\! \left( \!\! 1 \sminus 2 \! \Vsb \! \right) \!\!  \Omega_{d}^{\sminus1} \!
    \sminus \! \Vsb \!\! \left( \! 1 \sminus \! \Vsb \! \right) \!\!\! \left[ \Omega_{d}^{\sminus1} \! \Omega_{h} \Gm_{\Gm}\!\!\left( \! \Vsb \! \right) \! + \! \Gm_{\Fm}\!\!\left( \! \Vsb \! \right) \! \Omega_{g}\Omega_{d}^{\sminus1} \!\! \right] \! \right] \!\! \vphantom{\Big |}\right) \!\! = \! 0 \! . \nonumber 
	\eeqn
	Multiplying the left by $\det\left(\Gm_{\Fm}^{\sminus1}\left(\Vsb\right)\right)$ and the right by $\det\left(\Gm_{\Gm}^{\sminus1}\left(\Vsb\right)\right)$ and then re arranging yields
	\beqn
%	&&\nqquad\nqquad \det\left(\vphantom{\Big|}  \left[ \left(1-2\Vsb \right) \left[\Vsb\Omega_{g} +\Psi\right]\Omega_{d}^{\sminus1} \left[\Vsb\Omega_{h} +\Omega_{d} + \Psi\right]
%	\right.\right. \nonumber \\
%	&& \nquad -\Vsb\left(1-\Vsb\right) \left[\left[\Vsb\Omega_{g} +\Psi\right]\Omega_{d}^{\sminus1}\Omega_{h} \right.\nonumber \\
%	&& \left.\left.\left.+  \Omega_{g}\Omega_{d}^{\sminus1}\left[\Vsb\Omega_{h} +\Omega_{d} + \Psi\right] \right] \right]\vphantom{\Big |}\right) = 0,\nonumber  \\
%	%
%	&&\nqquad\nqquad \det\left(\vphantom{\Big|}  \Vsb^2\left[ \Omega_{g} \Omega_{d}^{\sminus1}\Omega_{h} + \Psi \Omega_{d}^{\sminus1}\Omega_{h} + \Omega_{g}\Omega_{d}^{\sminus1}\left(\Omega_{d}+\Psi\right) \right] \right. \nonumber \\
%	&& \nquad \left.  +2\Vsb \Psi \Omega_{d}^{\sminus1} \left(\Omega_{d}+\Psi\right) -  \Psi \Omega_{d}^{\sminus1}\left(\Omega_{d}+\Psi\right) \vphantom{\Big |}\right) = 0, \nonumber \\
%	%
	&&\nqquad\nqquad \det\left(\vphantom{\Big|}  \Vsb^2\left[ \Omega_{g} \Omega_{d}^{\sminus1}\left[\Omega_{g} + \Psi\right] + \Psi \Omega_{d}^{\sminus1}\Omega_{h} \right] \right. \nonumber \\
	&& \nquad \left.  +2\Vsb \Psi \Omega_{d}^{\sminus1} \left(\Omega_{d}+\Psi\right) -  \Psi \Omega_{d}^{\sminus1}\left(\Omega_{d}+\Psi\right) \vphantom{\Big |}\right) = 0. \label{eq:det-before-inverse}
	\eeqn
	The above is a quadratic in the scalar, $\Vsb$. %
	However, it should be noted that due to the determinant this equation is actually of order $Nm$ in $\Vsb$. %
	To see this, note that $\Vsb$ is pulled out of the determinant as follows
	\beqn
	&&\nqquad\nqquad \Vsb^{Nm}\det\left(\vphantom{\Big|} - {{1}\over{\Vsb^2}} \Psi \Omega_{d}^{\sminus1}\left(\Omega_{d}+\Psi\right)  + 2{{1}\over{\Vsb}} \Psi \Omega_{d}^{\sminus1} \left(\Omega_{d}+\Psi\right)   \right. \nonumber \\
	&&  \left.  +  \left[ \Omega_{g} \Omega_{d}^{\sminus1}\left[\Omega_{g} + \Psi\right] + \Psi \Omega_{d}^{\sminus1}\Omega_{h} \right] \vphantom{\Big |}\right) = 0.\nonumber 
	\eeqn
	Additionally, by assumption $\Vsb\neq 0 $, therefore dividing by $-\Vsb^{Nm}$ gives
	\beqn
	&& \nqquad \nqquad \det\left(\vphantom{{{1}\over{\Vsb^2}}}\right. \left({{1}\over{\Vsb^2}} - 2{{1}\over{\Vsb}} \right) \underbrace{\left[  \Psi \Omega_{d}^{\sminus1}\left(\Omega_{d}+\Psi\right) \right]}_{\Hm}    \nonumber \\
	&&     -  \underbrace{\left[ \Omega_{g} \Omega_{d}^{\sminus1}\!\! \left[\Omega_{g} + \Psi\right] + \Psi \Omega_{d}^{\sminus1}\!\! \Omega_{h} \right]}_{\Tm} \!\!{\left. \vphantom{{{1}\over{\Vsb^2}}} \right)}\!\! = 0.  \nonumber
	\eeqn
	From the above it is seen that the matrix $\Hm $ is positive definite and diagonal. Additionally, with some minor manipulation it is seen that $\Tm$ is symmetric. %
	It is stated in~\cite[16.51(c)]{seber} that there exists a matrix $\Sm\in \R^{Nm \times Nm}$ such that $\Sm^{\bf T}\Hm\Sm ={\bf I}$ and $\Sm^{\bf T} \Tm\Sm= \Lambda$. Therefore, multiplying the left and right by $\det\left(\Sm^{\bf T}\right)$ and $\det\left(\Sm\right)$ respectively gives
	\beqn
	\det\left(\vphantom{{{1}\over{\Vsb^2}}} {{1}\over{\Vsb^2}}{\bf I} - 2{{1}\over{\Vsb}} {\bf I}   -  \Lambda \right)\!\! = 0,  \label{eq:det-diag}
	\eeqn
	where $\Lambda$ is the diagonal matrix where all diagonal entries, $\lambda_i$, are the eigenvalues of the matrix $\Tm \Hm^{\sminus1}$. The above is equivalent to
	\beqn
	\prod_{i=1}^{Nm}\left( {{1}\over{\Vsb^2}} - 2{{1}\over{\Vsb}} -  \lambda_i \right)\!\! = 0.  \label{eq:prod-quad}
	\eeqn
	Note that~(\ref{eq:prod-quad}) is $Nm$ polynomials on $\Vsb$ with solutions
	\beqn
	\Vsb^D_{2i-1} &=&  {{1}\over{1 + \sqrt{1+\lambda_i}}},\nonumber \\ 
	\Vsb^D_{2i} &=&  {{1}\over{1 - \sqrt{1+\lambda_i}}},  \label{eq:quad-sol}
	\eeqn
	where $\Vsb_i^{D}$ is the $i$-th solution of~(\ref{eq:det-0}). %
	Additionally,~(\ref{eq:quad-sol}) corresponds to~(\ref{eq:det-sol}). This concludes the proof
\end{proof}